\documentclass[sigplan,10pt]{acmart}

\usepackage[utf8]{inputenc}

\usepackage{tikz}
\usepackage{xspace}
\usepackage{placeins}
\usepackage{enumitem}

\usepackage{graphicx}
\usepackage{subcaption}

\usepackage{algorithm}
\usepackage[noend]{algpseudocode}
\usepackage[cppcommentstyle,continuouslinenumbers,mainfont=small]{smartalgorithmic}

\definecolor{darkgreen}{rgb}{0.0, 0.4, 0.0}
\definecolor{crimson}{rgb}{0.4, 0.05, 0.1}
\usepackage[apply,nonotes]{xnotes}
\AddXNotesUser{at}{Andrei}{darkgreen}
\AddXNotesUser{s}{Sasha}{red}
\AddXNotesUser{d}{Daniel}{blue}
\AddXNotesUser{b}{Balaji}{purple}
\AddXNotesUser{todo}{TODO}{crimson}

\newcommand{\sasha}{\snote}

\usepackage{cleveref}

\NewDeclarationType{Function}{\textsf}
\NewDeclarationType{Variable}{\mathit}
\NewDeclarationType{MessageType}[1]{\textbf{``#1''}}
\NewDeclarationType{EntryReason}{\textsf}
\NewDeclarationType{Timer}{\textsf}
\NewDeclarationType{Constant}{\text}
\NewDeclarationType{Operation}{\mathit}

\ifdefined\blindedSubmission
  \newcommand{\BabyRaptr}{\textnormal{\textsc{Baby Raikou}}\xspace}
  \newcommand{\Raptr}{\textnormal{\textsc{Raikou}}\xspace}
  \newcommand{\BoldBabyRaptr}{\textnormal{\textbf{{\textsc{Baby Raikou}}}}\xspace}
  \newcommand{\BoldRaptr}{\textnormal{\textbf{{\textsc{Raikou}}}}\xspace}
  \newcommand{\AptosPlus}{\textnormal{Aptos+}\xspace}
  \graphicspath{{assets/private/}}
\else
  \newcommand{\BabyRaptr}{\textnormal{\textsc{Baby Raptr}}\xspace}
  \newcommand{\Raptr}{\textnormal{\textsc{Raptr}}\xspace}
  \newcommand{\BoldBabyRaptr}{\textnormal{\textbf{{\textsc{Baby Raptr}}}}\xspace}
  \newcommand{\BoldRaptr}{\textnormal{\textbf{{\textsc{Raptr}}}}\xspace}
  \newcommand{\AptosPlus}{\textnormal{Aptos+}\xspace}
  \graphicspath{{assets/public/}}
\fi

\newcommand\Baseline{$\textsc{Jolteon}^*$}
\newcommand{\queuing}{block inclusion time\xspace}

\DeclareGlobalConstant{true}
\DeclareGlobalConstant{false}
\DeclareGlobalConstant*[SupMQuorumInline]{\ensuremath{\left\lceil \frac{n + f + 1}{2} \right\rceil}}

\DeclareGlobalOperation[aDeliver]{a\_deliver}
\DeclareGlobalOperation[aBcast]{\mathit{a\_bcast}}
\DeclareGlobalOperation{messages}
\DeclareGlobalOperation{chain}
\DeclareGlobalOperation{block}
\DeclareGlobalOperation{blockPrefix}
\DeclareGlobalOperation{rank}

\DeclareGlobalMessageType[mBatch]{batch}
\DeclareGlobalMessageType[mPropose]{propose}
\DeclareGlobalMessageType[mQCVote]{qc-vote}
\DeclareGlobalMessageType[mAdvanceRound]{advance round}
\DeclareGlobalMessageType[mCCVote]{cc-vote}
\DeclareGlobalMessageType[mTCVote]{tc-vote}
\DeclareGlobalMessageType[mPoAVote]{poa-vote}
\DeclareGlobalMessageType[mPoA]{poa}

\DeclareGlobalTimer[tRoundTimeout]{Round Timeout}
\DeclareGlobalTimer[tQCVote]{QC-Vote}

\DeclareGlobalVariable{extendPrefix}
\DeclareGlobalVariable{commitPrefix}
\DeclareGlobalVariable{extendRank}

\DeclareGlobalFunction{OnNewQC}
\DeclareGlobalFunction{TryAdvanceRound}
\DeclareGlobalFunction{QCVote}
\DeclareGlobalFunction{FormQC}
\DeclareGlobalFunction{CommitQC}
\DeclareGlobalFunction{Add}
\DeclareGlobalFunction{AvailablePrefix}
\DeclareGlobalFunction{GetPayload}
\DeclareGlobalFunction{VerifyEntryReason}
\DeclareGlobalFunction{VerifyQC}
\DeclareGlobalFunction{VerifyCC}
\DeclareGlobalFunction{VerifyTC}

\DeclareGlobalEntryReason[rFullQC]{FullQC}
\DeclareGlobalEntryReason[rCC]{CC}
\DeclareGlobalEntryReason[rTC]{TC}

\DeclareGlobalVariable{replicaId}
\DeclareGlobalVariable*[rCur]{\ensuremath{r_{\VariableStyle{cur}}}}
\DeclareGlobalVariable*[rTimeout]{\ensuremath{r_{\VariableStyle{timeout}}}}
\DeclareGlobalVariable{qc}
\DeclareGlobalVariable*[qcHigh]{\ensuremath{\qc_{\VariableStyle{high}}}}
\DeclareGlobalVariable*[qcCommitted]{\ensuremath{\qc_{\VariableStyle{committed}}}}
\DeclareGlobalVariable*[qcDelivered]{\ensuremath{\qc_{\VariableStyle{delivered}}}}
\DeclareGlobalVariable*[qcMin]{\ensuremath{\qc_{\VariableStyle{min}}}}
\DeclareGlobalVariable*[qcMax]{\ensuremath{\qc_{\VariableStyle{max}}}}
\DeclareGlobalVariable*[qcGenesis]{\ensuremath{\qc_{\VariableStyle{genesis}}}}
\DeclareGlobalVariable*[qcParent]{\ensuremath{\qc_{\VariableStyle{parent}}}}
\DeclareGlobalVariable*[genesisBlock]{\ensuremath{B_{\VariableStyle{genesis}}}}
\DeclareGlobalVariable{cc}
\DeclareGlobalVariable{tc}
\DeclareGlobalVariable{sn}
\DeclareGlobalVariable{lastQCVote}
\DeclareGlobalVariable{ccVoted}
\DeclareGlobalVariable{entryReason}
\DeclareGlobalVariable{reason}
\DeclareGlobalVariable*[leader]{\ensuremath{L}}
\DeclareGlobalVariable{payload}
\DeclareGlobalVariable{proposal}
\DeclareGlobalVariable{round}
\DeclareGlobalVariable{voteRound}
\DeclareGlobalVariable{qcRound}
\DeclareGlobalVariable{ccRound}
\DeclareGlobalVariable{prefix}
\DeclareGlobalVariable{votePrefix}
\DeclareGlobalVariable{qcPrefix}
\DeclareGlobalVariable{ccPrefix}
\DeclareGlobalVariable{maxRound}
\DeclareGlobalVariable{maxPrefix}
\DeclareGlobalVariable{subBlocks}
\DeclareGlobalVariable*[nSubBlocks]{K}
\DeclareGlobalVariable{vote}
\DeclareGlobalVariable{votes}
\DeclareGlobalVariable{qcVotes}
\DeclareGlobalVariable{ccVotes}
\DeclareGlobalVariable{tcVotes}
\DeclareGlobalVariable{hash}
\DeclareGlobalVariable{sender}
\DeclareGlobalVariable{votePrefixes}
\DeclareGlobalVariable{voteData}
\DeclareGlobalVariable{signature}
\DeclareGlobalVariable*[storageRequirement]{\ensuremath{S}}

\DeclareGlobalVariable{batch}
\DeclareGlobalVariable{batches}
\DeclareGlobalVariable{poa}
\DeclareGlobalVariable{poas}
\DeclareGlobalVariable{poaVotes}
\DeclareGlobalVariable{myBatches}
\DeclareGlobalFunction{OnNewBlock}
\DeclareGlobalFunction{FetchQCData}
\DeclareGlobalFunction{VerifyPoA}

\DeclareGlobalVariable*[tEnter]{\ensuremath{t_{0}}}
\DeclareGlobalVariable*[rLast]{\ensuremath{r_{\VariableStyle{last}}}}

\DeclareGlobalVariable*[quorumSize]{\SupMQuorumInline}

\newcommand{\roundTimeoutDuration}{\ensuremath{(4+\epsilon)\Delta}}
\newcommand{\qcVoteTimerDuration}{\ensuremath{\epsilon\Delta}}

\newcommand{\myparagraph}[1]{\smallskip\subparagraph*{\textbf{#1.}}}

\newcommand{\TODO}{\textcolor{red}{\textbf{\smaller{TODO}}}\xspace}

\definecolor{sky}{RGB}{50,150,255}
\definecolor{crimson}{RGB}{200,20,60}

\algnewcommand{\Message}[1]{\ensuremath{\langle}#1\ensuremath{\rangle}}
\algnewcommand{\Multicast}[1]{\textbf{multicast }\Message{#1}}
\algnewcommand{\Send}[2]{\textbf{send }\Message{#1}\textbf{ to }{#2}}
\algnewcommand{\IfThen}[2]{\textbf{if} {#1} \textbf{then} {#2}}
\algnewcommand{\IfThenElse}[3]{\textbf{if} {#1} \textbf{then} {#2} \textbf{else} {#3}}

\newcommand{\tnot}{\textbf{not}~}
\newcommand{\tand}{\mathbin{\textbf{ and }}}
\newcommand{\tor}{\mathbin{\textbf{ or }}}

\algblockdefx{Parameters}{EndParameters}{\textbf{parameters:}}{}
\algtext*{EndParameters}

\algblockdefx[Function]{Function}{EndFunction}[2]{\textbf{function} #1(#2):}{}
\algtext*{EndFunction}

\algblockdefx[Operation]{Operation}{EndOperation}[2]{\textbf{operation} #1(#2):}{}
\algtext*{EndOperation}

\algblockdefx[Upon]{Upon}{EndUpon}[1]{\textbf{upon} #1:}{}
\algtext*{EndUpon}

\algblockdefx[UponStart]{UponStart}{EndUpon}{\textbf{upon start of the protocol}:}{}
\algtext*{EndUpon}

\algblockdefx[UponReceiving]{UponReceiving}{EndUpon}[2]{\textbf{upon receiving} \Message{#1} \textbf{from} #2:}{}
\algtext*{EndUpon}

\algblockdefx[UponTimerExpires]{UponTimerExpires}{EndUpon}[1]{\textbf{upon} #1 \textbf{timer expires}:}{}
\algtext*{EndUpon}

\algblockdefx[Variable]{Variable}{EndVariable}{\textbf{Variables:}}{}
\algtext*{EndVariable}

\algblockdefx[Objective]{Objective}{EndObjective}{\textbf{Objective:}}{}
\algtext*{EndObjective}

\algblockdefx[ForConstraint]{ForConstraint}{EndForConstraint}[1]{$\forall #1:$}{}
\algtext*{EndForConstraint}

\algblockdefx[Constraint]{Constraint}{EndConstraint}{\textbf{Constraints:}}{}
\algtext*{EndConstraint}

\algblockdefx[Where]{Where}{EndWhere}{\textbf{Where:}}{}
\algtext*{EndWhere}

\newcounter{statementcnt}
{
    \stepcounter{statementcnt}%
    \begin{myframe}{Problem statement {\arabic{statementcnt}}\ifthenelse{\equal{#1}{}}{}{ (#1)}}%
}{
    \end{myframe}%
}

\newlength{\dhatheight}

\newcommand{\ass}{{\sf AS}}

\newcommand{\psign}{\ensuremath{{\sf PSign}}}
\newcommand{\pver}{\ensuremath{{\sf PVer}}}
\newcommand{\ver}{\ensuremath{{\sf Ver}}}
\newcommand{\comb}{\ensuremath{{\sf Comb}}}

\newcommand{\pk}{{\sf pk}}
\newcommand{\sk}{{\sf sk}}
\newcommand{\cM}{\ensuremath{{\mathcal M}}}

\newtheorem{definition}{Definition}

\setcopyright{none}
\renewcommand\footnotetextcopyrightpermission[1]{} 

\begin{document}

\title{\BoldRaptr: Prefix Consensus for Robust High-Performance BFT}
\author{Andrei Tonkikh}
\authornote{Equal Contribution.}
\affiliation{
  \institution{Aptos Labs}
  \country{}
}
\author{Balaji Arun}
\authornotemark[1]
\affiliation{
  \institution{Aptos Labs}
   \country{}
}
\author{Zhuolun Xiang}
\affiliation{
  \institution{Aptos Labs}
  \country{}
}
\author{Zekun Li}
\affiliation{
  \institution{Aptos Labs}
  \country{}
}
\author{Alexander Spiegelman}
\affiliation{
  \institution{Aptos Labs}
  \country{}
}

\begin{abstract}

In this paper, we present {\Raptr}--a Byzantine fault-tolerant state machine replication (BFT SMR) protocol that combines strong robustness with high throughput, while attaining near-optimal theoretical latency. {\Raptr} delivers exceptionally low latency and high throughput under favorable conditions, and it degrades gracefully in the presence of Byzantine faults and network attacks.

Existing high-throughput BFT SMR protocols typically take either \emph{pessimistic} or \emph{optimistic} approaches to data dissemination: the former suffers from suboptimal latency in favorable conditions, while the latter deteriorates sharply under minimal attacks or network instability. {\Raptr} bridges this gap, combining the strengths of both approaches through a novel Prefix Consensus mechanism.

We implement {\Raptr} and evaluate it against several state-of-the-art protocols in a geo-distributed environment with 100 replicas.
\Raptr achieves 260,000 transactions per second (TPS) with sub-second latency under favorable conditions, sustaining 610ms at 10,000 TPS and 755ms at 250,000 TPS. It remains robust under network glitches, showing minimal performance degradation even with a 1\% message drop rate.
 
\end{abstract}

\maketitle

\pagestyle{plain} 

\section{Introduction}\label{sec:intro}

State machine replication (SMR)---the abstraction of a single infallible machine ``in the sky'' implemented on top of many fault-prone machines (replicas) connected by a network---is the holy grail of distributed systems. In particular, Byzantine fault-tolerant (BFT) SMR forms the basis of modern blockchain systems.

Most existing BFT State Machine Replication (SMR) systems adopt a leader-based approach.
In this paradigm, the protocol proceeds in numbered rounds (or ``views''), with each round assigning a single replica as the leader.
The leader proposes a new block of transactions to extend the ever-growing chain, and the other replicas vote on the proposal and commit the block after reaching consensus.

This paradigm enables an optimal good-case latency of 3 message delays under low load~\cite{pbft,abraham2021good,kuznetsov2021revisiting}.
However, if implemented naïvely, a single replica--i.e. the leader--is responsible for broadcasting the entire block, which limits throughput to the leader’s outbound bandwidth divided by the total number of replicas.
As the system scales, this bottleneck causes throughput to degrade proportionally, making the approach unsuitable for the high-throughput demands of modern large-scale blockchain systems.


Motivated by the need for robust, high-throughput, and low-latency Byzantine Fault Tolerant (BFT) protocols for blockchain systems, both academia and industry have increasingly explored designs based on directed acyclic graphs (DAGs)~\cite{aleph, allyouneed}.
These protocols offer a promising alternative to traditional leader-based designs through their leaderless DAG structure, which enables parallel data dissemination, mitigates the leader bottleneck, and achieves higher throughput.
Their asynchronous architecture further enhances resilience to network disruptions and replica failures, making them particularly well-suited for decentralized environments.
In particular, early systems such as Narwhal/Tusk\cite{narwhaltusk} and Bullshark\cite{bullshark} demonstrate up to $100\times$ throughput improvements over leader-based protocols in large-scale deployments.


The major drawback of early DAG-based BFT systems, however, was their significant latency overhead compared to the fastest leader-based protocols such as PBFT~\cite{pbft} and its derivatives.
Subsequent DAG-based protocols sought to improve latency~\cite{spiegelman2024shoal, shoal++, Sailfish, mysticeti}, but none achieved optimal theoretical latency while preserving the high throughput and robustness of earlier designs.
For instance, Mysticeti~\cite{mysticeti} reduced latency substantially by omitting the node certification step during DAG construction.
However, this came at the cost of robustness: even a 0.05\% message loss can increase latency by $10\times$~\cite{shoal++}, due to the need to fetch missing DAG nodes on the critical path.

This paper introduces {\Raptr}, a novel BFT system designed to simultaneously achieve robustness, high throughput, and optimal theoretical latency.
{\Raptr} builds on a variation of the Jolteon protocol~\cite{jolteon} to optimize latency, while leveraging parallel data dissemination to boost throughput.
Its core innovation lies in integrating the low-latency benefits of leader-based protocols with the high throughput of DAG-based BFT systems—without compromising resilience to adversarial failures or adverse network conditions.


\subsection{Technical Overview}\label{sec:intro:overview}

%
We measure ordering latency as the time since a transaction is submitted to a replica until that replica receives an ordering confirmation. 
For leader-based protocols,
this includes transaction dissemination latency, \queuing for transactions to be included in a leader's proposal, consensus latency for leader's proposal to be ordered, and post-processing that includes mempool notification and fetching transactions (if necessary). 

\paragraph{Baseline}
Our starting point is the latest open source version of the Aptos BFT protocol, which is a latency-improved version of the leader-based Jolteon~\cite{jolteon} consensus protocol, called \Baseline~\cite{jolteon-star}, equipped with a parallel data dissemination, called Quorum Store~\cite{quorumstore}.
\begin{itemize}[leftmargin=*]
    \item {\Baseline}~\cite{jolteon-star} reduces the common-case consensus latency of Jolteon~\cite{jolteon} from 5 to 3 message delays via a PBFT-style all-to-all voting~\cite{jolteon-star}. 
    The \queuing 
    for \Baseline equals that of Jolteon at one message delay on average, since leaders propose blocks every two message delays.

    \item Quorum Store~\cite{quorumstore} is a reliable-broadcast-based parallel data dissemination~\cite{narwhaltusk} to remove the leader bandwidth bottleneck and improve throughput. Each replica (in parallel) broadcasts a batch of transactions and the rest of the replicas persist the batch and send a signature back. When the sender receives a quorum of signatures, it aggregates them into a proof of availability (PoA) and broadcasts it, ensuring future retrieval of the batch. 
    Therefore, leaders can safely propose blocks 
    of\atadd{ small} PoA certificates\atremove{ (constant size)}
    instead of the full transactions data. 
\end{itemize}

As a result, the average ordering latency is 7 message delays: 3 for disseminating and certifying data, 
1 on average for \queuing, 
and 3 (the lower bound) for committing the block in consensus.\footnote{\atadd{Technically, there is also a small ``batch inclusion latency'', as part of the dissemination, equal, on average, to half of the batch creation interval. For simplicity, we ignore it in the theoretical discussion, but it will appear in the evaluations in \Cref{sec:eval}.}}

\paragraph{\BoldBabyRaptr}
The Quorum Store~\cite{quorumstore} approach removes the throughput bottleneck in leader-based consensus systems, allowing them to achieve throughput that is comparable to state-of-the-art DAG BFT protocols~\cite{mysticeti,shoal++}. 
However, Quorum Store introduces a latency of three message delays for data dissemination and certification prior to consensus, matching the overhead of consensus itself.
To reduce this cost, we first propose a naive optimistic Quorum Store variant, which we call {\BabyRaptr}.
The core idea is to merge the Quorum Store with the consensus logic. As before, all replicas broadcast their batches in parallel.
However, in {\BabyRaptr}, leaders propose batch metadata without waiting for a proof of availability (PoA).
To ensure data availability, replicas vote for a proposal only if they already possess the corresponding batch data locally. Otherwise, they must retrieve the missing data on the critical path before voting.
This design eliminates two message delays in the common case, reducing the latency to a single message delay.
However, it can suffer from high latency in the presence of faults or partial network disconnections, as data fetching becomes blocking.

\paragraph{\BoldRaptr}
Our main contribution is a novel consensus technique that addresses the problem of fetching data on the critical path.
The core idea is to allow replicas to vote on partial blocks.
Rather than waiting to fetch missing data, replicas may vote on a prefix of the proposed block consisting of only the batches they have locally.
The key safety property that underpins this design is \emph{prefix containment}.
Since different replicas may see different subsets of votes, they may commit different prefixes of the same proposed block.
However, as long as all committed prefixes are nested, safety is preserved by ensuring that the next round’s proposal extends the highest committed prefix observed so far.
For data availability (i.e., liveness), it suffices that at least one honest replica has the full data for the committed prefix.

In summary, the paper presents two key algorithmic innovations: (1) \emph{background data certification} integrated with the consensus mechanism, and (2) a novel \atreplace{sub-block}{prefix} voting scheme that allows \emph{partial block commits}.  


\subsection{Implementation and Evaluation}
We implemented {\BabyRaptr} and {\Raptr} within the publicly available Aptos codebase and evaluated them against a modified variant of the Aptos BFT protocol, as well as two state-of-the-art DAG-based BFT protocols: Shoal++~\cite{shoal++} and Mysticeti~\cite{mysticeti}. 
Our implementation incorporates several system-level optimizations, including network-level separation for consensus and data messages, efficient aggregate signatures, and a reputation mechanism based on heuristics.

Our evaluation shows that under favorable conditions, {\Raptr} can commit up to 260,000 transactions per second with sub-second latency in a 100-node geo-distributed deployment{\textemdash}achieving roughly $2\times$ the throughput of the state-of-the-art.
Furthermore, \Raptr consistently delivers low and stable latency throughout the full load spectrum: 610 ms at 10,000 TPS and 755 ms at 250,000 TPS.

In terms of robustness, we show that \Raptr experiences minimal to no performance degradation under partial or network-wide glitches, while existing protocols degrade significantly.

{\Raptr} bridges the gap between pessimistic and optimistic data dissemination approaches---achieving the high throughput of Aptos BFT (which uses a pessimistic Quorum Store) while matching the near-optimal latency of {\BabyRaptr} (which relies on optimistic batch proposals).

\section{Preliminary}
\label{sec:prelim}

\subsection{System Model} \label{sec:prelim:model}

We consider a system consisting of $n$ replicas: $\Pi=\{p_1, \dots, p_n\}$ connected by a network that allows any pair of replicas to exchange messages.
%
%
A malicious adversary can corrupt any of the replicas at any point in time, as long as the total number of corrupted replicas does not exceed $f$ where $n=3f+1$.
Once a replica is corrupted, the adversary can observe all of its internal state and takes full control over it.
All non-corrupted replica follow the protocol as specified.
A replica is called \emph{malicious} if it is corrupted by the adversary at any time during its execution, and \emph{honest} otherwise.

All replicas know the full \emph{membership} of the protocol, which consists of the identities of all participants and the meta-information of constant size per participant as required by the protocol.
This meta-information includes the cryptographic public keys, allowing the \atreplace{nodes}{replicas} to establish secure authenticated peer-to-peer channels.
For simplicity, in this paper, we consider a system with fixed membership.
However, the standard state machine replication reconfiguration techniques apply.

We assume a \emph{partially synchronous} network, that is, there must exist a publicly known communication delay upper bound $\Delta$ and an unknown moment in time called \emph{Global Stabilization Time (GST)} such that any message between two honest \atreplace{nodes}{replicas} sent after the GST is delivered and processed by the recipient within $\Delta$ time units after it has been sent.
Within these bounds, the network is controlled by the adversary. 
It can choose when GST occurs, drop or delay any messages sent before GST, and delay messages sent after GST within the upper limit of $\Delta$.

The adversary and all \atreplace{nodes}{replicas} are assumed to be computationally bounded.
More specifically, we assume that they cannot break any of the cryptographic primitives used in the protocol.

We define {\em common case} as the scenario in which the network is synchronous (after GST) and all replicas are honest.


\subsection{Definitions}\label{sec:prelim:def}
We focus on the Byzantine Atomic Broadcast (BAB) problem. 
\begin{definition}[Byzantine Atomic Broadcast] \label{def:bab}
Every replica $p\in\Pi$ can broadcast messages by calling $\aBcast(m, \sn)$. 
Every replica $p\in\Pi$ can output $\aDeliver(m, \sn, q)$, where $m$ is a message, $\sn \in \mathbb{N}$ is a sequence number, and $q\in \Pi$ is the replica that called the corresponding $\aBcast(m, \sn)$. 
The Byzantine atomic broadcast guarantees the following properties:
\begin{itemize}[leftmargin=*]
    \sloppypar
    \item Totality: If an honest replica $p$ outputs $\aDeliver(m, \sn, q)$, then every honest replica $p'$ eventually outputs $\aDeliver(m, \sn, q)$.

    \item Non-Duplication: For each sequence number $\sn \in \mathbb{N}$ and replica $q \in \Pi$, an honest replica $p$ outputs $\aDeliver(m, \sn, q)$ at most once regardless of $m$. 

    \item Validity: If an honest replica $q$ calls $\aBcast(m, \sn)$, then every honest replica $p$ eventually outputs $\aDeliver(m, \sn, q)$.

    \item Total order: If an honest replica $p$ outputs $\aDeliver(m, \sn, q)$ before $\aDeliver(m', \sn', q')$, then no honest replica $p'$ outputs $\aDeliver(m', \sn', q')$ before $\aDeliver(m, \sn, q)$.
\end{itemize}
\end{definition}

In our paper, we let message $m$ be a batch of blockchain transactions.



\subsubsection*{Latency}
To define latency metrics, we define an additional property that most Byzantine Atomic Broadcast protocols satisfy, including our protocol.
\begin{definition}[Block Proposal]
    The Byzantine Atomic Broadcast protocol has replicas proposing block proposals consisting of transactions, explicitly or implicitly~\footnote{For leader-based BAB protocols, the leader of a given round can explicitly propose a block proposal extending the previous proposal. 
    For DAG-based BFT protocols, there are also chosen leaders that can implicitly propose a block proposal which consists its proposed DAG node and all causally dependent DAG nodes that are not included in the previous implicit block proposal.}.
    To propose a block $B$, a replica multicasts $\langle \mPropose, B \rangle$.
\end{definition}

\begin{definition}[Latency Metrics]
For a message $m$ (or a transaction in the blockchain context), we define the following latency metrics.
    \begin{itemize}[leftmargin=*]
        \item Dissemination latency. \atrev{The duration from the moment an honest replica $p$ calls $\aBcast(m, \sn)$ until some honest replica multicasts $\langle \mPropose, B \rangle$ such that $B$ includes $m$ (directly or indirectly).}
        \item Consensus latency. The time duration from the first honest replica multicasts $\langle \mPropose, B \rangle$, to the last honest replica outputs $\aDeliver(m, \sn, p)$, \dadd{where $B$ includes $m$ (directly or indirectly)}.
        \item Ordering latency. The ordering latency is equal to the dissemination latency plus the consensus latency.
    \end{itemize}
\end{definition}

The dissemination latency includes the {\em \queuing}, which measures the delay for a transaction to be included in a block proposal. The \queuing is, \atreplace{in expectation}{on average}, half of the block time (duration between two consecutive block proposals).

Looking ahead, as  discussed in~\Cref{sec:raptr:performance}, {\Raptr} has a dissemination latency of 2 message delays (including \queuing of 1 message delay), and a consensus latency of 3 message delays (which is optimal~\cite{kuznetsov2021revisiting,abraham2021good}).
Thus, {\Raptr} has an ordering latency of 5 message delays.

\subsection{Cryptography Primitives}
\label{sec:prelim:crypto}

\begin{definition}[Non-Interactive Aggregate Signature]
\label{def:NIAS}
Let $\{\pk_1,\ldots,\pk_n\}$ be a vector of public key shares, and $\{\sk_1,\ldots,\sk_n\}$ be a vector of secret key shares.
Consider $n$ signers where the $j$-th signer has $(\pk,\{\pk_i\}_{i\in \atreplace{[n]}{\Pi}},\sk_j)$.
A non-interactive aggregate signature scheme $\ass$ for a finite message space $\cM$ and $n$ signers is a tuple of polynomial-time algorithms $\ass=(\psign,\pver,\comb,\ver)$ defined as follows:
\begin{enumerate}[leftmargin=*]
    %
    \item $\psign(\sk_i, msg)\rightarrow \sigma_i:$ The partial signing takes as input a secret key share $\sk_i$, and a message $msg\in \cM$. It outputs a signature share $\sigma_i$. 
    \item $\pver(\pk_i,msg,\sigma_i,)\rightarrow 0/1:$ The partial signature verification takes as input a public key share $\pk_i$, a message $msg$, and a signature share $\sigma_i$. It outputs 1 (accept) or 0 (reject).
    \item $\comb(\{\, (i, \sigma_i) \,\}_{i \in S})\rightarrow \sigma/\bot:$ The combine algorithm takes as input a set $S$ of signers and their signature shares $\{\sigma_i\}_{i\in S}$. It outputs either a signature $\sigma$ or $\bot$.
    \item $\ver(\{(\pk_i, msg_i)\}_{i\in S}, \sigma)\rightarrow 0/1:$ The signature verification algorithm takes as input a set of tuples consisting public key shares $\pk_i$ and messages $msg_i$ of each signer $i$ in the set $S$, and a signature $\sigma$. It outputs 1 (accept) or 0 (reject).
\end{enumerate}
\end{definition}


Looking ahead, in the {\Raptr} protocol, each replica will sign for each block its hash and the prefix number of the received batches, i.e., $m=(H(B), \emph{prefix})$. 
Our implementation uses a simplified version of no-commit proofs~\cite{giridharan2021no} to reduce verification overhead\atadd{ and we group batch digests into ``sub-blocks'' to limit the number of possible prefixes}---see \Cref{sec:system:signature} for more details.

\section{Baseline: Aptos Consensus}\label{sec:baseline}
Our starting point and baseline for evaluation comparison is an existing consensus protocol currently deployed on the Aptos Blockchain~\cite{aptos-codebase}, which has two major components---{\em Quorum Store}~\cite{quorumstore} for data dissemination and {\Baseline}~\cite{jolteon,jolteon-star} for block ordering. 
%
{\Baseline} is a latency-improved version of Jolteon~\cite{jolteon} that reduces the common-case consensus latency from 5 to 3 message delays via a PBFT-style all-to-all voting~\cite{jolteon-star}.
Quorum Store~\cite{quorumstore} is a reliable-broadcast-based parallel data dissemination~\cite{narwhaltusk} to remove the leader bandwidth bottleneck and improve throughput. 


\subsection{Notations}\label{sec:baseline:notation}
We define a few notations used in BFT consensus protocols that are standard in the literature.
\begin{itemize}[leftmargin=*]
    \item {\em Rounds and leaders}. The protocol advances in rounds, such that in each round $r$ there is a designated replica $\leader_r$ that is the \emph{leader} in round $r$.
    \footnote{The mapping between rounds to leaders can be static (e.g., round robin) or dynamic (updated according to replicas' participation).
    }

    \item {\em Blocks}. A block $B$ contains a \emph{payload} of transactions, a round number $r$, a quorum certificate (QC), and an optional timeout certificate (TC). 
    
    \item {\em Quorum certificate}. A valid quorum certificate (QC) $\qc$ contains a valid aggregated signature $\qc.\signature$ on a string $\qc.\hash$. We say that a QC $\qc$ \emph{references} a block $B$ if $\qc.\hash = H(B)$, where $H$ is a cryptographic hash function. We define the round $\qc.\round$ to be the round of Block $B$.\atremove{ In {\Baseline}, we compare QCs by their rounds and say that $\qc_1 > \qc_2$ iff $\qc_1.\round > \qc_2.\round$.}
    
    \item {\em Timeout certificate}. A timeout certificate (TC) aggregates a quorum of round timeout messages. A valid TC $\tc$ contains an aggregated signature $\tc.\signature$ on the round number $r$ in which it was formed, and a vector $V$ 
    \dadd{containing the round numbers of the highest QCs that replicas attached to their respective timeout messages}.
    We denote $r$ as $\tc.\round$ and \emph{tc.highest\_qc\_round} to be the highest round in $V$. 
\end{itemize}
In Section~\ref{sec:raptr}, we will extend the Block, QC, and TC notations to describe the {\Raptr} protocol.

\subsection{{\Baseline}}\label{sec:baseline:protocol}
Below, we informally present the {\Baseline} protocol.
Although we later formally present {\Raptr} in a full standalone description, the informal description of {\Baseline} helps build up the intuition of \Raptr's correctness and design choices.

The goal of the {\Baseline} protocol is for replicas to agree on an ever-growing chain of blocks $B_0,B_1,...$ such that each block $B_i$, $i>0$, contains, among other things, a QC $\qc$ that references Block $B_{i-1}$. The block $B_0$ contains the predefined genesis QC.

\paragraph{Common case protocol}
In {\Baseline}, a replica advances to round $r$ once it sees a valid \emph{EntryReason}, which is either a QC or a TC from round $r-1$.
When the leader $L_r$ of round $r$ advances to round $r$, it forms a block $B$ that contains the highest QC it has ever seen and an \emph{EntryReason} to advance to round $r$. 
The leader $L_r$ of round $r$ broadcasts its block $B$ to all other replicas.

When a replica receives a block $B$ from $L_r$ in round $r$ for the first time, it first advances to round $r$ ({if it has not done so yet}) by verifying \emph{EntryReason} as follows:
\begin{itemize}[leftmargin=*]
    \item if \emph{EntryReason} is a QC $qc$: verify that $\emph{qc.round} = r-1$. 
    \item if \emph{EntryReason} is a TC $tc$: verify that $\emph{tc.round} = r-1$ and $\emph{B.qc.round} \geq \emph{tc.highest\_qc\_round}$. This means that the TC was formed in the previous round and the QC in the proposed block is at least as high as \emph{tc.highest\_qc\_round}, which is the highest round of all the QCs attached to the timeout messages that form the TC.
\end{itemize}
Then, it signs $H(B)$ and broadcasts its \emph{QC-vote} to all.
Upon collecting a quorum of \emph{QC-votes}, replicas aggregate a QC $qc$ that references $B$ and broadcast a commit certificate vote (\emph{CC-vote}) to all replicas.
In addition, they advance to round $r+1$, and $L_{r+1}$ also broadcasts its block for round $r+1$.
Each replica that collects a quorum of \emph{CC-votes} or receives an aggregated commit certificate (CC) can commit~$B$.

\paragraph{Failure case} When replicas do not receive a block from $L_r$ or enough QC-votes to form a QC in time (e.g., $L_r$ is faulty or the network experiences disconnections), replicas broadcast a timeout message that contains the highest round QC they have seen so far. 
When receiving a quorum of timeout messages, replicas form a TC and advance to the next round without committing a block in the current round. 

\paragraph{Performance}
In the common case, the consensus latency of {\Baseline} is 3 message delays: one message delay to propose a block, one message delay to broadcast a QC-vote, and another message delay to broadcast a CC-vote.
The ordering latency is five message delays: one message delay for replicas to transmit their transactions to the next leaders, one message delay (in expectation) for transactions to be included in a proposed block (\queuing), and three message delays for consensus latency. 
The extra message delay for \queuing comes from the fact that a new block is proposed in {\Baseline} every two message delays, which means that on average (under the uniform arrival distribution) each transaction waits one message delay for the proposal. 
{\Baseline} has a quadratic message complexity.

\subsection{Quorum Store}\label{sec:qs}

A recent breakthrough in BFT throughput was the realization that data must be disseminated in parallel.
For example, DAG-based consensus protocols~\cite{narwhaltusk} were able to improve throughput over state-of-the-art leader-based consensus protocols by $50\times$ with their leaderless DAG structure, in which all validators disseminate nodes in parallel.


Quorum Store~\cite{quorumstore}, currently deployed in the Aptos network to scale {\Baseline}, adopts the core idea of leaderless data dissemination to provide a parallel dissemination layer that can be integrated with any leader-based consensus protocol.


The protocol is symmetric, where all validators repeat the following in parallel:
\begin{itemize}[leftmargin=*]
    \item Create a batch of transactions and broadcast them to all validators.
    \item Upon receiving a batch, persist it to storage, sign its hash, and send the signature back to the {batch author}.
    \item Upon receiving a quorum of signatures for a batch, aggregate them to form a \emph{proof-of-availability { (PoA)}} and broadcast the proof to all validators. 
\end{itemize}

\subsection{Putting Things Together}


When enhanced with the Quorum Store layer, {\Baseline} includes only a set of \atreplace{availability proofs}{proofs of availability}, not the full transaction data.
Each proof guarantees that a quorum of honest validators has stored the corresponding data, allowing it to be fetched later (for example, before execution) if a validator does not have it locally.
Since the proofs are significantly smaller than the raw data, the block sizes shrink substantially and the leader bandwidth is no longer a bottleneck.
As a result, Quorum Store significantly improves the throughput of {\Baseline} while preserving its robustness.

\paragraph{Performance}
Theoretically, the quorum store protocol requires 3 message delays, linear communication per sender, and quadratic communication overall, while the {\Baseline} protocol has a quadratic message complexity. 
The average ordering latency is 7 message delays: a dissemination latency of 3 message delays,
an average \queuing of 1 message delay,
and a consensus latency of 3 message delays.
Compared to {\Baseline} without Quorum Store, the ordering latency increases by 2 message delays, as a trade-off for higher throughput.
We present the results of the evaluation of this protocol in~\Cref{sec:eval}.

\section{\BoldBabyRaptr} \label{sec:baby}

Although Quorum Store scales the throughput of {\Baseline} by an order of magnitude~\cite{quorumstore}, it also increases its ordering latency by 2 message delays.
In this paper, we merge Quorum Store into the {\Baseline} protocol in a way that eliminates the latency penalty of Quorum Store while maintaining its high throughput and robustness.

We start here by describing {\BabyRaptr}, which achieves {\em high throughput without sacrificing latency} in the common case. In the next section, we describe the full {\Raptr} protocol that adds resilience to bad networks and faulty validators.

The core idea is to integrate \atreplace{part of the Quorum Store logic}{data availability certification} directly into the consensus protocol while\atadd{ also} running the original Quorum Store asynchronously, off the critical path.
As before, all validators broadcast batches of transactions, and the quorum store protocol proceeds asynchronously in the background to generate proofs of availability. 
To save two message delays, \atrev{leaders in {\BabyRaptr} additionally propose digests of the batches they received from the replicas {\em without} waiting for the quorum store proofs.}

To integrate the \atreplace{proofs of availability generation on}{data availability certification with} the consensus protocol, replicas sign proposals as part of the consensus protocol only if they locally persist the batches behind the included digests (this is similar to the Quorum Store logic).
Therefore, a QC that refers to Block $B$ also serves as \atreplace{a proof of availability}{a data availability proof} for the \atreplace{batch digests in $B$}{batches referenced in $B$}. 
Thus, if block $B$ is committed, then it is guaranteed that all its data is available in the system\atremove{, which is exactly the guarantee that quorum store provides}.

While {\BabyRaptr} reduces ordering latency by two message delays, it introduces a key drawback: replicas cannot sign a consensus proposal unless they locally possess all the data referenced by the block's digests.
If any data is missing, replicas must fetch it on the critical path, incurring additional latency or triggering leader timeouts—since data retrieval takes at least two message delays.
This situation can arise in the presence of poor network conditions or slow/malicious replicas.
The issue is exacerbated when batch creators behave maliciously, for example, by sending batches exclusively to the leader.\footnote{A similar vulnerability exists in uncertified DAG BFT protocols such as Mysticeti\cite{mysticeti}.}

In the next section, we describe how {\Raptr} elegantly addresses this issue. 
However, a practical approach to mitigate this issue is to revert to the baseline version.
Since Quorum Store runs in the background, each leader can locally decide to go back to proposing proofs whenever it observes performance degradation.
Moreover, whenever a proof for a digest is available, replicas should always include it instead of just the digest.
In practice, each block can be a combination of proofs and digests that depend\atadd{s} on local heuristics. 
System considerations are discussed in~\Cref{sec:system}.





\section{\BoldRaptr}\label{sec:raptr}

In this section, we present the {\Raptr} protocol. 
In \S\ref{sec:raptr:intuition}, we provide an overview of the protocol, and in \S\ref {sec:raptr:description}, we give a full description.
Due to space limitations, the full pseudocode and correctness proofs are deferred to \Cref{sec:pseudocode,sec:proof}, respectively.

\subsection{Protocol Intuition}\label{sec:raptr:intuition}

\begin{figure*}[t!]
    \centering
    \includegraphics[width=\linewidth]{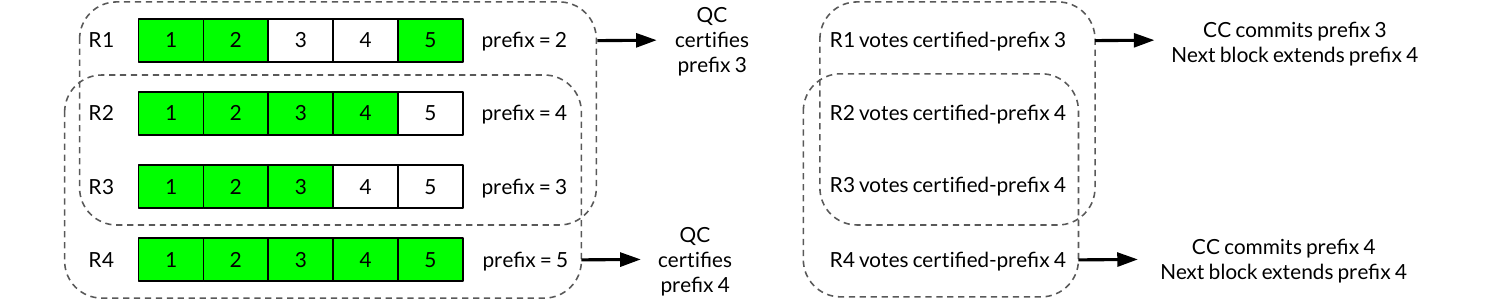}
    \caption{Illustration of the non-binary voting on prefixes in {\Raptr}, as described in Section~\ref{sec:raptr:intuition} and~\ref{sec:raptr:description}. Four replicas ($R_1,\dots,R_4$) receive different subsets of batches (green) of the same block, and vote on the longest available prefix. A quorum of QC-votes form a quorum certificate (QC), certifying a block's prefix. In this example, suppose $S=2$, QC-votes from $R_1,R_2,R_3$ certify prefix 3, while QC-votes from $R_2,R_3,R_4$ certify prefix 4. After forming a QC, replicas then vote {to commit the certified prefix}. A quorum of CC-votes forms a commit certificate (CC), committing a prefix. Here, CC-votes from $R_1,R_2,R_3$ commit prefix 3, and those from $R_2,R_3,R_4$ commit prefix 4. The next block proposal extending either CC will extend the maximum certified prefix in the quorum, which is 4 in both case.}
    \label{fig:prefix}
\end{figure*}

The {\Raptr} protocol enhances resilience by addressing the limitations of {\BabyRaptr}, while preserving its latency benefits.
{\BabyRaptr} removes Quorum Store proof-of-availability certification from the critical path to reduce ordering latency, but becomes vulnerable to slow or malicious batch creators that fail to distribute batches to replicas.
The core issue lies in {\em binary} block voting, a mechanism common to most BFT consensus protocols: replicas vote either for the entire block or not at all.
In {\BabyRaptr}, a replica will only vote once all referenced batches are available; if even a single batch is missing, it refuses to vote.

{\Raptr} overcomes this limitation by introducing {\em non-binary} voting on \atreplace{sub-blocks}{prefixes}.
Instead of requiring full availability of all batches in a block before voting, replicas vote on the longest available prefix of batches they have received. This ensures that consensus continues to make progress without stalling on missing batches.
The challenge in the {\Raptr} design lies in achieving consensus with non-binary votes. 

Conceptually, {\Raptr} reaches agreement on an ever-growing sequence of batches, rather than on the specific blocks proposed by each leader.
Each leader still proposes a block of batches that extends the previously committed batch sequence, but replicas—due to differing local views—may commit different \emph{prefixes} of these batches.
The safety of {\Raptr} is ensured by the prefix containment property: as long as the committed prefixes at different replicas extend one another, the protocol maintains agreement.

To enforce safety, {\Raptr} has the following key constructions. \Cref{fig:prefix} provides an illustration.

\begin{itemize}[leftmargin=*]
    \item In each round, along with the proposed block digest, replicas vote on the longest available prefix of batches they have received. A quorum certificate (QC) contains a quorum of votes for the same block and can certify the availability of a prefix of batches referenced in the block.
    %
    %
    We pick an \emph{availability requirement} parameter $\storageRequirement \ge f+1$ and say that a QC \emph{certifies} prefix $k$ if at least $S$ replicas voted for prefix $k$ or larger.
    It is easy to see that all batches in the certified prefix are retrievable as at least $S-f \ge 1$ honest replicas have them stored locally.

    \item Then, replicas vote to commit a QC prefix. Each vote includes the prefix certified by the QC, and the QC itself to justify the vote.
    The aggregate signature from a quorum of such votes commits the {\em minimum} prefix attached to the votes.
    We refer to it as a \emph{commit certificate}.
    
    \item 
    {If progress stalls, a replica can timeout, stop voting in this round, and send a timeout message with its highest QC.}

    \item The next-round leader enters the new round in three ways:
    \begin{enumerate}[leftmargin=*]
        \item The best case is when the next-round leader obtains a \emph{full-prefix} QC (one that certifies all batches proposed in the block).
        Then, it propose\atadd{s} a new block extending that QC, without having to wait for the commit votes.
        
        \item In case the next round leader fails to get the full-prefix QC, it enters the new round once it obtains a commit certificate.
        It then extends the \emph{maximum} prefix attached to the commit certificate votes.

        \item Finally, the worst case that can only happen in case of network asynchrony or a faulty leader is when the previous round is timed out. Upon collecting a quorum of timeout messages, the next-round leader
        proposes a block extending the \emph{maximum} prefix certified by the QCs in the timeouts.
    \end{enumerate}

    
\end{itemize}

\subsection{Protocol Description}\label{sec:raptr:description}


This section provides a high-level description of {\Raptr}. 
Figure ~\ref{fig:prefix} illustrates the logic of voting on and committing block prefixes.  

{\Raptr} progresses in rounds, starting from round 1, which extends the genesis block. 
A replica advances to a new round when it has a valid entry reason as defined by the protocol. 
For clarity, we describe valid entry reasons at the end of this section.
\atrev{A key distinction from classic pipelined leader-based protocols is that, instead of forming a chain of blocks, {\Raptr}
directly constructs a sequence of batches, built out of \emph{block prefixes}.}


    \paragraph{Entering a New Round} A replica resets local timers and sends its reason to the leader, allowing the leader to advance as well. If the replica is the leader, it proposes a block extending the highest {\em certified prefix} (defined below) it knows. 
    The block proposal includes the reason and the batches' metadata (pulled from the in-memory storage). The leader then broadcasts the block to all participants.
    
    \paragraph{Receiving a Block} Upon receiving the first block from the round-$r$ leader, a replica performs validity checks, stores the block, processes the attached quorum certificate (QC) (QC definition described below), and attempts to advance the round based on the reason contained in the block. The replica also starts its voting timer.
    
    \paragraph{Voting on a Block (QC-vote).} A replica may 
    vote on a block in its current round up to two times:
    \begin{enumerate}[leftmargin=*]
        \item When the voting timer expires after receiving the block\atadd{, if not yet QC-voted by the next rule}. 
        \item When all batches referenced in the block have been received.
    \end{enumerate}
    Each vote contains a signature on a message containing the vote type, block hash, round number, and the longest received batch prefix. Note that the timer is a configurable system parameter designed to allow replicas to receive additional batches, enabling them to vote on higher prefixes.
    In principle, we could have allowed the replicas to issue more than two QC-votes without violating safety, but it is unclear if that would result in any benefits for the protocol.

    \paragraph{Forming a Quorum Certificate (QC)} When a replica collects a quorum of QC-votes, it forms a QC, which {\em certifies} a prefix of batches\atremove{ (a sub-block)}. This certified prefix is computed as the $\storageRequirement$'th maximum prefix in the quorum of QC-votes, where $\storageRequirement \geq f+1$, ensuring that the certified prefix is retrievable from $\geq \storageRequirement -f \geq 1$ honest replicas. QCs are ranked by the (round, certified prefix) tuple. 
    {Each replica may form a QC up to two times:}
    \begin{enumerate}[leftmargin=*]
        \item If no QC has yet been formed or received for this round or a higher one.
        \item If the QC certifies all batches in the block (full-prefix QC).
    \end{enumerate}
    Similarly to QC-vote, we could have allowed \atreplace{to form}{forming} a QC more than two times without any safety violations.
    
    \paragraph{Voting on a QC (CC-vote)} Upon forming or receiving a QC, a replica updates its highest QC and issues a vote for the prefix certified by the QC (after checks). The vote includes the vote type, block hash, round number,\atadd{ and the} certified prefix. 
    The replica sends the corresponding QC with the CC-vote for recipients to verify the certified prefix. 
    
    \paragraph{Forming a Commit Certificate (CC)}
    \atrev{When a replica collects a quorum of CC-votes for the same block, it forms a commit certificate (CC). A valid CC commits the minimum certified prefix from the votes it contains and, recursively, all block prefixes in the chain of parents. Committing a block prefix means committing all of the proofs of availability in the block and the prefix of batches whose digest are included without the proofs of availability.}
    
    \paragraph{Round Timeout} When a round timeout expires, replicas broadcast signed timeout messages with their highest QCs. A replica that gathers a quorum of timeout votes forms a timeout certificate (TC).

    \paragraph{Valid reasons for entering a new round} A replica enters the next round whenever it has a valid reason, which can be:
    \begin{enumerate}[leftmargin=*]
        \item A full-prefix QC. The next block extends the full previous block.
        \item A valid commit certificate (CC). The next block extends the maximum certified prefix in the votes.
        \item A valid timeout certificate (TC). Similarly to the CC case, the next block extends the maximum prefix certified by the QCs in the voters. 
    \end{enumerate}

\subsection{Correctness intuition}\label{sec:raptr:correctness}
Formal proofs can be found in Appendix~\ref{sec:proof}. Here we provide high-level correctness intuition.
The safety of the protocol relies on the prefix containment property---the committed prefixes at different replicas extend each other.
To ensure this, recall that {\Raptr} commits the minimum prefix in a quorum of commit certificate votes and extends the next round proposal on the maximum prefix in a quorum of commit certificate/timeout votes.
Due to quorum intersection, any valid proposal must extend the committed prefix.
Additionally, it is always safe to extend a full-prefix QC\atadd{ of the previous round} as no larger-prefix QC could possibly exist and, therefore, no larger-prefix QC could possibly be committed in that round.

\subsection{Performance Analysis}\label{sec:raptr:performance}
We analyze the latency of \Raptr in the common case. 
The consensus latency spans 3 message delays: one each for the proposal, QC vote, and CC vote.
The dissemination latency spans 2 message delays: one to send the transaction batch, and another for \queuing (as the protocol proposes block every two message delays).
Therefore, the ordering latency of \Raptr is 5 message delays.

\paragraph{Decoupling of availability and safety quorums}~\label{sec:raptr:decouple}
A key intuition behind \Raptr's high throughput under high load lies in its \atreplace{sub-block}{prefix} voting design. Under high load, data dissemination becomes a bottleneck, and not all \atreplace{nodes}{replicas} may receive all the data of the block (full block) during voting. Traditional BFT protocols require at least a quorum ($2f+1$) of \atreplace{nodes}{replicas} to have the full block for it to be committed. In contrast, \atreplace{\Raptr relaxes this requirement: only $\storageRequirement \geq f+1$ \atreplace{nodes}{replicas} need the full block to generate QCs certifying it. Once a quorum of \atreplace{nodes}{replicas} votes for such QCs, the full block can be committed.}{in {\Raptr}, a full-prefix QC can be formed with only $\storageRequirement \geq f+1$ votes for the full block, and the rest of the votes in the quorum ($2f+1 - \storageRequirement$) can be partial.} 
This distinction becomes evident in~\Cref{fig:failure-free-performance} of~\Cref{sec:eval}, where we compare the peak throughput of \BabyRaptr and \Raptr.

\section{System Considerations} 
\label{sec:system}

In this section, we present the system-level contributions that complement and enhance \Raptr's performance in real-world deployments.

\subsection{Reducing \queuing}
\label{sec:system:queuing}

\Raptr creates a block every two message delays with a \queuing of one message delay.
However, optimistic proposal techniques~\cite{doidge2024moonshot} can propose blocks every message delay, thus reducing the \queuing to half a message delay.
Alternatively, similar to Shoal++~\cite{shoal++}, $K$ instances of \Raptr can be run in parallel, with each instance offset by $2/K$ message delays.
This can effectively reduce the \queuing by a factor of $K$, as long as the instances are aligned such that there is an equal interval between the block proposals of two consecutive instances.
The evaluation omits this optimization, as both techniques are orthogonal to \Raptr and can be integrated independently.



\subsection{Aggregate signatures and sub-blocks}
\label{sec:system:signature}

In \Raptr, replicas can vote on different prefixes, resulting in non-identical QC votes, rendering efficient multisignature schemes unsuitable for verification.
Instead, efficient aggregate signature schemes are necessary to ensure that signature verification does not become a bottleneck.
Below are concrete realizations of such schemes (\Cref{def:NIAS}).

\paragraph{BGLS aggregate signatures} In this scheme~\cite{boneh2003aggregate}, each signer holds a tuple of public and secret key shares. Both the signature share and the aggregate signature have a constant size. Verification of partial signatures requires $O(1)$ pairing operations. However, the aggregate signature verification takes $O(M)$ pairing operations, where $M$ denotes the maximum number of \atreplace{sub-blocks allowed in a block}{possible prefixes}.\atremove{ This aggregate signature scheme, while straightforward, was very expensive in our experiments and emerged as a bottleneck for fast consensus.}

\paragraph{No-commit proofs~\cite{giridharan2021no}} The no-commit proofs scheme below improves the aggregate signature verification cost to $O(1)$ pairing operations, with the expense of using more public and secret key shares. This scheme is based on a virtualization of multi-signatures which require the same message to be signed and is an aggregate signature scheme tailored for the use case where the messages signed by each replica consist of a common message (\atadd{e.g., the }block hash) and \atreplace{a distinct message}{distinct \emph{tags}} of bounded size (\atreplace{prefix number of received batches}{e.g., the prefixes}). 

The signers will sign the common message using different secret key shares corresponding to the \atreplace{distinct message}{tag}. More specifically, each signer holds $K$ tuples of public and secret key shares, each corresponding to a\atadd{ possible}\atremove{ batch} prefix number~\footnote{\atreplace{The no-commit proof has optimizations}{The original scheme of~\cite{giridharan2021no} includes an optimization} to reduce the number of keys by signing with respect to the binary representation of the prefix number. We simplify the implementation by omitting the optimization, since\atadd{ it slightly slows down the verification and,} in our use case, the number of sub-blocks is small.}. Each signer will sign the block hash using the secret key share according to its received\atremove{ batch} prefix number $l_i$, e.g., for a signer \atreplace{that receives batches $1,2,3$ but not $4$, it will sign the block hash using its third secret key share}{voting for block $B$ with prefix $3$, it will sign $H(B)$ with its third secret key}. Signature verification will be evaluated based on the claimed prefix numbers and corresponding public key shares. As a result, \atreplace{signature verification has $O(1)$ cost, which is similar to multi-signatures}{the signature verification cost is the same as for a simple multi-signature}.


For the implementation of \Raptr, we use the no-commit proofs~\cite{giridharan2021no} for our aggregate signature scheme to reduce the signature verification overhead.

\begin{atreview}
\paragraph{Sub-blocks} As both schemes discussed above are sensitive to the number of possible prefixes, we reduce this number to a constant by grouping the batch digests into $\nSubBlocks$ \emph{sub-blocks}.
Thus, the prefix is a number between $0$ and $\nSubBlocks$, where $0$ indicating that only the proofs of availability are to be committed and $\nSubBlocks$ indicating that the full block is to be committed.
\end{atreview}

\subsection{Network Optimizations}
\label{sec:system:network}

BFT systems commonly employ full-mesh connectivity, where each replica maintains a TCP or QUIC connection with every other replica to enable direct message exchange.

In \Raptr, replicas continuously transmit messages of varying types and sizes. Quorum Store batch messages, which carry raw transaction payloads, can be up to hundreds of kilobytes in size. In contrast, non-data and consensus messages are only at most a few kilobytes. This imbalance makes message prioritization critical for performance.

When small consensus messages share a connection with large batch messages, they experience delays from both transmission overhead and queuing at the receiver.
To avoid this, \Raptr uses three separate TCP connections per replica pair: one each for quorum store batch messages, quorum store non-data messages, and consensus messages. This separation ensures that large data transfers do not interfere with latency-sensitive traffic.

In existing DAG-based protocols~\cite{narwhaltusk,bullshark,spiegelman2024shoal,shoal++,mysticeti}, payload dissemination is tightly coupled with DAG construction (therefore with consensus), making network-layer separation inherently challenging. 
While certified DAG protocols\cite{bullshark,spiegelman2024shoal,shoal++} may permit some decoupling, their consensus messages remain large---often containing a linear number of node certifications from the previous round---reducing the effectiveness of the network-layer separation. 
In contrast, \Raptr decouples batch dissemination from consensus, allowing both to proceed asynchronously.
This decoupling also enables batches to be produced at a higher rate than consensus for better bandwidth utilization.
As a result, \Raptr can fully leverage dedicated network channels to isolate data and consensus traffic, improving performance by reducing contention and latency.



\subsection{Reputation}
\label{sec:system:reputation}

\paragraph{Leader reputation}
Non-Byzantine failures such as crash failures and partial network partitions are the most common form of failures in today's Blockchain systems. 
Therefore, these systems benefit from having a reputation system that de-prioritizes failed replicas as proposers~\cite{cohen2022aware}, reducing timeout opportunities, and ensuring progress in every round.
Briefly, such reputation systems track the proposal status and voters for a window of rounds on which every replica has a common view, and use this information to deterministically pick leaders every round. 
Previously failed leaders get fewer proposal opportunities proportional to the number of failures in the window. 

\paragraph{Batch creator reputation}
Similar to proposers, certain batch creators can be less reliable than others. A batch creator reputation system helps proposers only pick batches from reputed creators that are able to disseminate their batches to the replicas in a timely fashion. Since proposers are rewarded for proposing successfully, proposing batches that are missing would trigger QC-vote timeouts leading to prefix votes. This can reduce the block rate and the number of transactions ordered per round, increasing consensus latency. 

The Batch creator reputation works as follows. Whenever a replica timer QC-votes a prefix or issues a Timeout message, it attaches a bitmap of authors of the missing batches. This bitmap is agreed upon as part of the QC or TC certification process. To avoid Byzantine behavior, an author is blamed only if at least $f+1$ of the replicas blame it. The batches from blamed authors are only included after their quorum store availability proof is available.

Similarly, replicas also blame slow authors whose batches are unavailable at the time of proposal receipt.
Block proposers use this information to reorder \atreplace{sub-blocks}{batch digests}, placing slow authors toward the end of the block, helping \Raptr's prefix \atreplace{commitment}{commit} mechanism.

\paragraph{Minimum batch age}
Transmitting and processing large batch messages in the quorum store can take longer than transmitting and processing small consensus proposal messages. When the proposal arrives, the batch may be missing locally at the replicas, delaying QC-votes. 
Furthermore, triangle inequality violations~\cite{lumezanu2009triangle} in networks can induce the same effect: the delay between the batch creator and the replica is longer than the delay between the batch creator and the proposer plus the proposer to replica hops. 
To avoid waiting on the critical path of consensus, we configure a minimum batch age before the batch can be included in a proposal. This ensures that the batches have had a sufficient head start to disseminate through the system, so when consensus proposes the batches, the replicas have the batches locally and can cast their votes without waiting.


\section{Evaluation} 
\label{sec:eval}

We evaluate \Raptr against state-of-the-art protocols and answer the following questions through our analysis:
\begin{enumerate}[leftmargin=*]
    \item What is the performance of \Raptr under ideal conditions?
    \item What is the breakdown of \Raptr's latency compared to the Baseline?
    \item How does \Raptr behave under faulty conditions?
\end{enumerate}

\paragraph{Implementation}
We implemented \Raptr and \BabyRaptr by modifying the publicly available, production-grade source code of the Aptos blockchain\footnote{\url{https://github.com/aptos-labs/aptos-core}}.
The codebase includes an implementation of the Aptos Baseline protocol (\S\ref{sec:baseline}), which facilitated our development.
We adapted the code to run only the consensus components, disabling execution and other blockchain-specific modules.
This isolation enables accurate measurement of consensus performance without interference from other subsystems.
Furthermore, as described in \S\ref{sec:system:network}, we modified the network stack to support multiple TCP connections between replicas and routed specific message types to dedicated connections. 
The source code is available on Github.\footnote{\url{https://github.com/aptos-labs/aptos-core/tree/raptr-public}}

\begin{figure*}[!h]
    \centering
    \includegraphics[width=\textwidth]{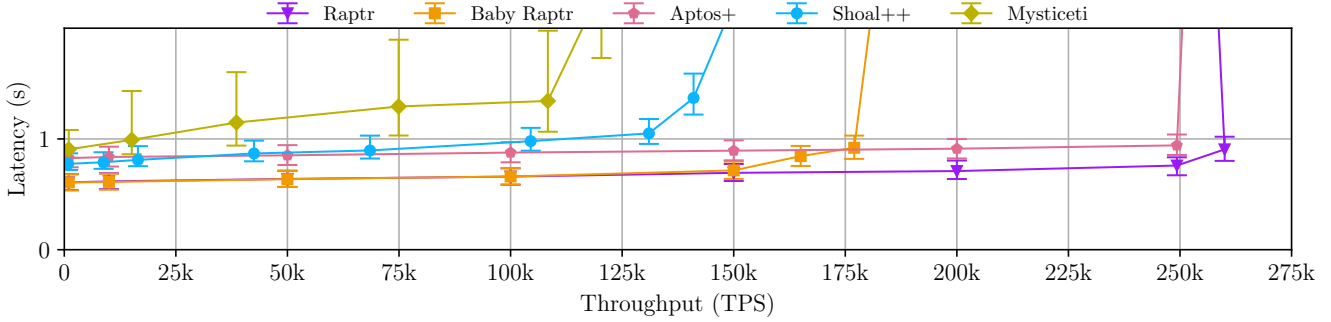}
    \vspace{-24pt}
    \caption{Common case performance of \Raptr versus other protocols. The points represent the 50th percentile latency and the error bars show the 25th and 75th percentile latencies respectively.}
    \label{fig:failure-free-performance}
\end{figure*}

\paragraph{Protocols under Test}
We compare the performance of \Raptr and \BabyRaptr against Shoal++~\cite{shoal++}, Mysticeti~\cite{mysticeti}, and \AptosPlus.
Shoal++ is a recent state-of-the-art certified DAG-based BFT protocol, and its prototype is available in the Aptos repository\footnote{\url{https://github.com/aptos-labs/aptos-core/tree/paper-shoal}}.
Mysticeti is an uncertified DAG-based protocol deployed in production by the Sui blockchain~\cite{sui}, and we use its open-source prototype\footnote{\url{https://github.com/asonnino/mysticeti}} in our evaluation.
\AptosPlus corresponds to the Baseline protocol (\S\ref{sec:baseline}) used in the Aptos blockchain, enhanced with the multi-TCP connection optimization described in \S\ref{sec:system:network}.
Note that Autobahn~\cite{giridharan2024autobahn} also separates data dissemination from consensus, similar to the Baseline protocol, and is therefore expected to exhibit similar performance to \AptosPlus under fault-free conditions with a well-engineered implementation.

\paragraph{Experimental Setup}
\sloppypar
We deployed a geo-distributed testbed using Google Cloud Platform (GCP) to emulate a decentralized network at global scale.
Our setup consists of 100 replicas evenly distributed across 10 GCP regions worldwide:
3 in Asia (\texttt{asia-northeast3}, \texttt{asia-southeast1}, \texttt{asia-south1}).
1 each in South America (\texttt{southamerica-east1}), South Africa (\texttt{africa-south1}), and Australia (\texttt{australia-southeast1}),
2 in the US (\texttt{us-west1}, \texttt{us-east1}),
and 2 in Europe (\texttt{europe-west4}, \texttt{europe-southwest1}).
Round-trip times between regions range from 25ms to 317ms.
Each replica runs on a \texttt{n2d-standard-64} VM instance with 64 vCPUs and 256GB of memory.

We use a variable number of client threads to generate load at different target throughputs.
Each client connects to a single replica and issues requests in an open-loop fashion to reach the desired TPS.
For example, a client sending a request each millisecond generates a 1000 TPS workload.
Each request consists of a 310-byte transaction representing a peer-to-peer transfer transaction.

We measure end-to-end latency from the time a transaction enters the mempool until the transaction is ready for execution and the mempool can be notified of its ordering.
Client-to-replica latencies are excluded, as they remain constant across protocols.
Each experiment lasted five minutes to allow the system to reach steady-state behavior.

The quorum store is configured to produce batches every 150ms, with a maximum of 450 transactions per batch.
A minimum batch age of 30ms is enforced to ensure sufficient dissemination time across replicas, mitigating the effects of triangle inequality in real-world networks.
For \Raptr, the number of sub-blocks in a block is set to 8.

\subsection{Performance under fault-free conditions}


First, we measure the performance of all protocols under non-faulty conditions.
Figure~\ref{fig:failure-free-performance} plots end-to-end latency versus throughput.
\Raptr achieves a peak throughput of 260k TPS while maintaining sub-second latency.
At this load, each VM utilizes 1.44 Gbps of bandwidth. 
\AptosPlus scales similarly to \Raptr but begins to saturate around 250K TPS.
Quorum store-based protocols send batches more frequently and independently of round progression, thereby utilizing bandwidth more effectively and achieving higher performance.

In contrast, Shoal++ and Mysticeti exhibit sharp latency spikes beyond 130k TPS and 110k TPS, respectively, indicating early saturation.
Notably, both DAG-based protocols fail to maintain sub-second latency beyond 100k TPS.

\begin{figure}[t]
  \centering
  \begin{subfigure}[b]{\linewidth}
    \centering
    \includegraphics[width=0.85\linewidth]{plots/raptr-breakdown.pdf}
    \vspace{-6pt}
    \caption{\Raptr}
    \label{fig:latency-breakdown:raptr}
  \end{subfigure}
  \begin{subfigure}[b]{\linewidth}
    \centering
    \includegraphics[width=0.85\linewidth]{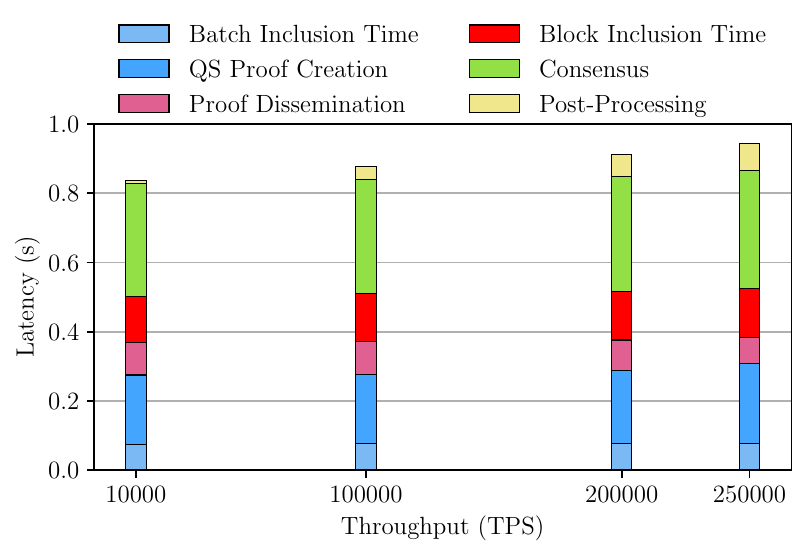}
    \vspace{-6pt}
    \caption{\AptosPlus}
    \label{fig:latency-breakdown-aptos}
  \end{subfigure}
  \vspace{-16pt}
  \caption{Latency breakdown for \Raptr and \AptosPlus under fault-free conditions.}
  \label{fig:latency-breakdown}
  \vspace{-24pt}
\end{figure}

\BabyRaptr shows similar behavior to \Raptr but saturates around 177k TPS.
At higher throughputs, delayed batch dissemination and processing postpone block voting.
Since \BabyRaptr couples the availability quorum with the BFT quorum (as discussed in~\S\ref{sec:raptr:decouple}), 
replicas that do not receive batches begin timing out, resulting in slower progress and early saturation.




\subsection{Performance Breakdown}

Figure~\ref{fig:latency-breakdown} shows a detailed breakdown of end-to-end latency for \Raptr and \AptosPlus, explaining how \Raptr achieves improved performance.
This is the breakdown of the results in Figure~\ref{fig:failure-free-performance}.
At low to moderate loads (up to 100K TPS), \Raptr achieves 25\% lower latency than \AptosPlus.
At 250K TPS---just before the saturation point---\Raptr maintains 755ms latency, a 20\% improvement over \AptosPlus.

The breakdown shows that \Raptr eliminates the Quorum Store proof creation step of 2 message delays, that is present in \AptosPlus, from the critical path, yielding significant latency improvements.
Consensus latency remains mostly stable across both protocols, though slightly higher for \Raptr at high TPS due to delayed batch arrivals, despite a fixed minimum batch age configuration.

Post-processing latency reflects the time required to materialize batches into transactions, which may involve waiting for payloads referenced by quorum store proofs.
This component is dominated by expensive operations such as cloning large vectors of transactions and sequentially notifying the mempool.
As throughput increases, this cost gradually rises with larger batches of transactions.

Finally, we observe that at higher TPS, the ratio of proofs in a block increases, with proofs occupying up to 22\% of the block space.
This occurs because larger batches arrive slightly late and miss the proposal, while the corresponding proofs become available before the next round.

\subsection{Faulty scenarios}

\begin{figure}[!hb]
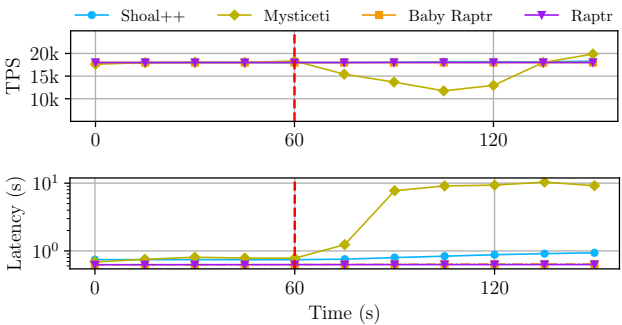

  \centering
  \includegraphics[width=\linewidth]{plots/drops1-timeline-tps.pdf}
  \includegraphics[width=\linewidth]{plots/drops1-timeline-lat.pdf}
  \caption{Impact of a partial network glitch on performance. Note the latency y-axis is in log-scale.}
  \label{fig:failures-drop1pct-5pctnodes}
\end{figure}

As mentioned in~\Cref{sec:system:reputation},\atremove{ \Raptr's}  leader reputation system can rotate out failed replicas, making experiments with permanently failed leaders less informative for demonstrating robustness. Therefore, we design our failure experiments based on patterns observed in real-world failure scenarios.

\paragraph{Partial network glitch}
We simulate a transient network fault by dropping 1\% of egress messages at 5\% of the replicas, resulting in a total drop rate of 0.05\% across the network.
Such scenarios can occur in practice when datacenters in specific regions experience intermittent network issues, leading to flaky TCP connections.
Such partial failures typically go undetected by reputation systems, which are designed to handle more persistent faults.

Figure~\ref{fig:failures-drop1pct-5pctnodes} shows the performance of all protocols over time under a low-to-moderate load of 18K TPS, with fault injection beginning at 60 seconds (marked by the red line).
The uncertified DAG protocol, Mysticeti, experiences significant degradation in both throughput and latency, as it performs critical-path synchronization of missing DAG nodes.
In contrast, the certified DAG protocol Shoal++ performs synchronization off the critical path and thus maintains better throughput. However, its latency increases by approximately 35\% post-fault injection, primarily due to round progress timeouts being triggered.
\Raptr and \BabyRaptr remain unaffected, maintaining the same performance as before, as a sufficient number of replicas receive batches in time to make optimistic progress without stalling.

\begin{figure}[t]
  \centering
  \includegraphics[width=\linewidth]{plots/dropsall-timeline-tps.pdf}
  \includegraphics[width=\linewidth]{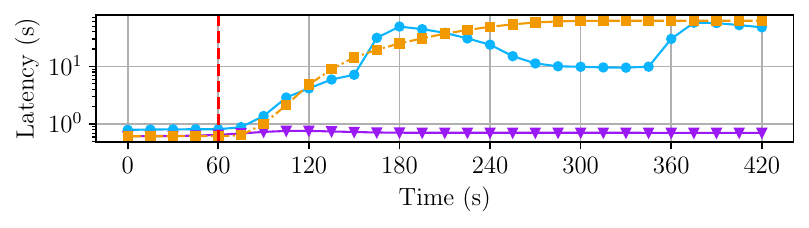}
  \includegraphics[width=\linewidth]{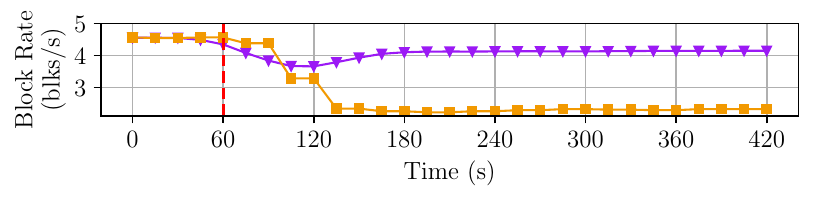}
  \caption{Impact of a full network glitch on performance. Note the latency y-axis is in log-scale.}
  \label{fig:failures-drop1pct-allnodes}
\end{figure}

\paragraph{Full network glitch}
We also evaluate performance under a network-wide transient fault, where all replicas drop 1\% of egress messages.
Figure~\ref{fig:failures-drop1pct-allnodes} shows the performance of \Raptr, \BabyRaptr, and Shoal++ over time under a low-to-moderate load of 18K TPS, with the fault injected at 60 seconds (marked by the red line).

Shoal++ and \BabyRaptr are noticeably affected. Shoal++ throughput exhibits a sinusoidal pattern as it intermittently fails to order anchors in some rounds, then compensates by committing a backlog of DAG nodes once an anchor is ordered.
This batching behavior causes latency to grow and fluctuate in sync with throughput.
\BabyRaptr's block rate drops to 2.5 blocks/second as batches are fetched in-line before voting, creating a backlog of pending transactions that reduces throughput and increases latency.

\Raptr is only mildly impacted. As forming the quorum certificates requires timer-induced QC-votes, the block rate drops slightly, increasing end-to-end latency by approximately 15\%, from 615ms to 710ms.
Despite this, \Raptr maintains flat TPS with no failed rounds. 
\Raptr's separation of the availability quorum from the BFT safety quorum differentiates it from \BabyRaptr (as discussed in \Cref{sec:raptr:decouple}) and allows it to commit full blocks in every round.

\subsection{Summary}
In summary, our evaluation demonstrates that:
\begin{itemize}[leftmargin=*]
\item \Raptr performs exceptionally well under ideal conditions. \Raptr combines the high throughput of pessimistic data dissemination (\AptosPlus) with the low latency of optimistic dissemination (\BabyRaptr).
\item \Raptr eliminates quorum store proof creation latency, achieving up to 25\% lower end-to-end latency compared to \AptosPlus.
\item \Raptr is robust to network-wide transient faults, maintaining stable throughput and bounded latency where other protocols degrade significantly.
\end{itemize}




\section{Related Work} 
\label{sec:related}

The BFT field was pioneered by the landmark PBFT protocol~\cite{pbft}, which provided the first practical solution to the Byzantine consensus problem. PBFT achieved the optimal latency of three message delays~\cite{kuznetsov2021revisiting,abraham2021good}, but its throughput was limited by the bandwidth constraints of a single replica. Since then, a substantial body of research has explored the trade-off between high throughput and low latency in BFT consensus protocols. Broadly, existing approaches can be categorized into leader-based and DAG-based protocols.

\paragraph{Leader-based BFT}
In an effort to increase throughput, numerous protocols~\cite{tendermint,buterin2017casper,hotstuff,jolteon,jalalzai2023fast,kang2024hotstuff,sui2022marlin,abspoel2020malicious,doidge2024moonshot} have integrated efficient leader rotation and pipelining, sometimes at the cost of higher consensus latency, such as HotStuff~\cite{hotstuff} and Jolteon~\cite{jolteon}.
Later, Moonshot achieves three-round consensus latency while retaining the desirable properties of its predecessors. Concurrently, Jolteon*~\cite{jolteon-star} was proposed to reduce Jolteon’s latency to three rounds and has been deployed in production in the Aptos blockchain.

A major breakthrough in throughput was achieved a few years ago with the realization that data dissemination must be parallelized~\cite{narwhaltusk, MirBFT, dispersedledger, arun2022dqbft}.
This insight spurred a line of research~\cite{neiheiser2021kauri,narwhaltusk,quorumstore,giridharan2024autobahn,MirBFT} focused on decoupling data dissemination from metadata ordering to improve throughput in leader-based consensus protocols.
\dadd{Kauri~\cite{neiheiser2021kauri} proposes tree-based communication abstraction for BFT consensus protocols to 
reduce leader's bandwidth bottleneck. 
However, the tree topology requires more network hops to disseminate data increasing the overall latency.
}
Narwhal~\cite{narwhaltusk} has a DAG-based mempool that can be integrated with any leader-based protocol to boost throughput by decoupling transaction dissemination from metadata ordering. 
Aptos productionizes Narwhal via Quorum Store~\cite{quorumstore}, which employs reliable broadcast~\cite{bracha1987asynchronous} to certify transaction batches rather than constructing a DAG. 
Autobahn~\cite{giridharan2024autobahn} is the state-of-the-art high-performance leader-based BFT protocol that is built on top of the Quorum Store data dissemination ideas - it links certified batches into chains, enabling faster failure recovery. In geo-distributed experiments with 4 \atreplace{nodes}{replicas}, it achieves nearly 250k transactions per second.
In comparison, \Raptr delivers similarly high throughput and maintains subsecond latency across 100 geo-distributed \atreplace{nodes}{replicas}.

\paragraph{DAG-based BFT}

Due to their structure, DAG-based protocols~\cite{aleph,allyouneed,hashgraph} naturally excel at efficient data dissemination.
Consequently, recent works~\cite{narwhaltusk,bullshark,bullsharksync,mysticeti,shoal++,Sailfish,bbca,xie2025fides,keidar2022cordial} have primarily focused on reducing latency while sustaining high throughput.
The state-of-the-art DAG-based BFT protocols, such as Cordial Miners~\cite{keidar2022cordial}\textbackslash Mysticeti~\cite{mysticeti}, Sailfish~\cite{Sailfish}, and Shoal++~\cite{shoal++}, have achieved low latency comparable to the state-of-the-art high throughput leader-based BFT protocols.
However, they lack robustness---as demonstrated in Section~\ref{sec:eval}, Mysticeti~\cite{mysticeti} exhibits latency spikes and throughput degradation even under minor network failures (1\% message drops at 5\% of \atreplace{nodes}{replicas}). Shoal++ can tolerate these smaller failures but suffers performance degradation when failures become more severe (1\% message drops at all \atreplace{nodes}{replicas}).

In fact, Mysticeti faces the same issue as {\BabyRaptr}. In {\BabyRaptr}, leaders optimistically propose the metadata of uncertified batches, and replicas cannot vote until they fetch all batches included in the proposal. Similarly, in Mysticeti, replicas cannot add a node to their local DAGs until they have fetched all the $n-f$ uncertified nodes from the previous round referenced by the node. {\Raptr} overcomes this issue in the leader-based approach by enabling \atreplace{voting on sub-blocks}{prefix voting} and partial block commits.

In terms of theoretical good-case latency, Mysticeti, Shoal++, and Raptr all achieve a four-message delay to commit. However, Raptr is faster in practice because DAG-based protocols experience longer per-message delays. This is due to the additional communication required for DAGs to advance rounds--intuitively, each of the  $n-f$ "leaders" must wait for $n-f$ messages from the previous round before proposing a node in the next round. Additionally, timeouts are sometimes enforced to ensure better DAG connectivity, further increasing message delay time.   

\paragraph{Flexible BFT quorums}
Many BFT consensus protocols explore adjusting quorum thresholds to offer flexible safety-liveness tradeoffs under varying network conditions~\cite{momose2021multi,malkhi2019flexible,xiang2021strengthened,neu2021ebb,lokhava2019fast,arun2020duobft}. In contrast, as discussed in~\Cref{sec:raptr:decouple}, \Raptr decouples the availability quorum from the safety quorum solely for performance, without affecting protocol correctness or aiming to provide flexible guarantees.

\section{Conclusion} 
\label{sec:conclusion}

{\Raptr} bridges the gap between fast optimistic data dissemination and robust pessimistic data dissemination by enhancing optimistic data dissemination with \atreplace{sub-block}{prefix} voting and partial block commits.
In a geo-distributed setting with 100 \atreplace{nodes}{replicas},
{\Raptr} sustains 250,000 TPS at just 755ms latency in a globally distributed deployment of 100 replicas under favorable conditions.
It also remains robust during network-wide glitches, exhibiting only minimal performance degradation.



\bibliographystyle{plain}
\bibliography{references}
\balance

\appendix

\clearpage
\nobalance

\section{Formal notation}\label{sec:notation}


\atadd{Let us introduce some formal notation that we will use in the pseudocode in \Cref{sec:pseudocode} as well as in the proof in \Cref{sec:proof}.}

\sloppypar
A \emph{\atreplace{sub-block}{block prefix}} is a tuple $(B, \prefix)$, where\atadd{ $B$ is a block and} $0 \le \prefix \le \nSubBlocks$.
A \emph{\atremove{sub-block }rank} is a tuple $(\round, \prefix)$. 
\atreplace{Sub-block ranks}{Ranks} are compared lexicographically, i.e., $(r_1, p_1) \le (r_2, p_2)$ iff $r_1 < r_2$ or $r_1 = r_2$ and $p_1 \le p_2$.
We use $\rank(\qc)$ to denote the\atremove{ sub-block} rank of a $\qc$, i.e., $\rank(\qc) = (\qc.\round, \qc.\prefix)$.
We say that $\qc_1$ is \emph{higher} (resp., \emph{lower}) than $\qc_2$ iff $\rank(\qc_1) > \rank(\qc_2)$ (resp., $\rank(\qc_1) < \rank(\qc_2)$).

We say that $\qc$ \emph{references the block} $B$ iff $\qc.\hash = H(B)$ and \emph{references the \atreplace{sub-block}{block prefix}} $(B, \prefix)$ iff $\qc.\hash = H(B)$ and $\qc.\prefix = \prefix$.
%
From the collision-resistance property of the hash function, we can reasonably assume that every QC references at most one block.\footnote{Technically speaking, hash collisions do exist, but we can ignore their existence as, with but negligible probability, no \atreplace{node}{replica}, whether honest or malicious, will\atremove{ ever} be able to find them within a reasonable time frame. Here and throughout the rest of the paper, we ignore such technicalities to maintain the readability.}
Moreover, since a valid QC contains an aggregate signature from $\quorumSize > f$ signatures (\cref{line:verify-qc:start}) and honest \atreplace{nodes}{replicas} only sign a QC if they \atreplace{have the}{seen the corresponding} block,
from the unforgeability of the aggregate signatures, we can also reasonably assume that a QC references at least one valid block.\footnote{Technically, it holds with but negligible probability for computationally-bounded adversaries.}
We use the notation $\block(\qc)$ and $\blockPrefix(\qc)$ to denote the block and the \atreplace{sub-block}{block prefix} referenced by $\qc$, respectively.

\sloppypar
$\chain(\qc)$ denotes the sequence of \atreplace{sub-blocks}{block prefixes} $(B_0, \prefix_0), \dots, (B_l, \prefix_l)$ that starts at $(B_0, \prefix_0) = (\genesisBlock, 0)$, ends at $(B_l, \prefix_l) = \blockPrefix(\qc)$, and
$\forall i < l: (B_i, \prefix_i) = \blockPrefix(B_{i+1}.\qcParent)$.
We use the notation $\messages(\chain(\qc))$ to refer to the sequence of messages contained in $\chain(\qc)$.

We say that $\qc_1$ is a \emph{prefix} of $\qc_2$ and write $\qc_1 \preceq \qc_2$ iff for some $(B_i, \prefix_i) \in \chain(\qc_2)$, $\block(\qc_1) = B_i$ and $\qc_1.\prefix \le \prefix_i$.
%
%
%
It is easy to see that if $\qc_1 \preceq \qc_2$, then $\messages(\chain(\qc_1))$ is a prefix of $\messages(\chain(\qc_2))$ (in the regular sense for sequences) and that the $\preceq$ relation is commutative.


We say that $\qc$ is \emph{directly committed} if any honest replica invokes $\CommitQC(\qc)$ (\cref{line:commit-qc:start}) at any point in time during the execution.

\section{Pseudocode}\label{sec:pseudocode}

\begin{atreview}
The detailed pseudocode of {\Raptr} can be found in \Crefrange{alg:raptr-first}{alg:raptr-last}. 
\begin{itemize}
    \item We start by defining local variables and timers as well as reciting the notation and the interface for the digital signatures in \Cref{alg:raptr:definitions}.

    \item We then describe the data structures and verification functions for the cryptographic certificates in \Cref{alg:raptr:structures-and-verification}. 

    \item \Cref{alg:raptr:quorum-store} describes the quorum store implementation (see \Cref{sec:qs}).

    \item \Cref{alg:raptr:aux} introduces three key auxiliary functions for {\Raptr}: $\TryAdvanceRound$, which is invoked whenever a new valid Entry Reason is observed, $\OnNewQC$, which is invoked whenever a new QC is formed or received, and $\CommitQC$, which is invoked whenever we commit something.
    It also defines the start of the protocol (\cref{line:start}) and processing of $\mAdvanceRound$ messages (\cref{line:advance-round:start})

    \item \Cref{alg:raptr:qc-votes} describes the processing of block proposals and the full process of QC-voting and forming new QCs.

    \item Finally, \Cref{alg:raptr:cc-and-tc} describes the processes CC-voting and TC-voting as well as forming the CCs and TCs.
\end{itemize}
\end{atreview}

\subsection{Implicit verification}

\atadd{To avoid excessive boilerplate, all the verification functions in \Cref{alg:raptr:structures-and-verification} are assumed to be invoked implicitly each time a corresponding certificate is received over the network.
Additionally, for the same reason, all the partial signatures are assumed to be verified with $\pver$ upon receiving (see \Cref{sec:prelim:crypto}).
We sometimes omit other instances of boilerplate verification of message fields to keep the pseudocode focused on the core logic of the protocol.
%
If any of the verifications fail, the message is ignored and the sender can be presumed Byzantine.}


\newcounter{raptrparts}
\setcounter{raptrparts}{0}

\begin{algorithm*}[p]
    \begin{smartalgorithmic}[1]
        

        \Multiline \textbf{variables:}
            \LineComment{\Raptr:}
            \State $\rCur = 0$ \tab\tab\tab\tab\tab\tab\tab\tab\Comment{The current round}
    
            \State $\rTimeout = 0$ \tab\tab\tab\tab\tab\tab\tab\tab\Comment{The last round this \atreplace{node}{replica} voted to time out}
    
            \State $\entryReason = \bot$ \tab\tab\tab\tab\tab\tab\Comment{The reason for entering $\rCur$}
    
            \State $\lastQCVote = (0, 0)$ \tab\tab\tab\tab\tab\tab\Comment{The last issued QC-Vote, as (round, prefix)}
            
            \State $\ccVoted = \text{map }[\round] \to \{\true, \false\}$ \Comment{Whether the replica already CC-voted in a round}
    
            \State $\qcHigh = \qcGenesis$ \tab\tab\tab\tab\tab\tab\Comment{The highest QC known to the \atreplace{node}{replica}}
        
            \State $\qcCommitted = \qcGenesis$ \tab\tab\tab\tab\tab\Comment{The highest QC committed by the \atreplace{node}{replica}}
    
            \State $\proposal = \text{map }[\round] \to \VariableStyle{block}$ \tab\tab\Comment{Map from a round number to the proposal received from the leader of the round}
                        
            \State $\qcVotes = \text{map }[\round][\hash][\replicaId] \to (\prefix, \signature)$ 

            \State $\ccVotes = \text{map }[\round][\replicaId] \to (\qc, \signature)$ 
            
            \State $\tcVotes = \text{map }[\round][\replicaId] \to (\qc, \signature)$

            \LineComment{Quorum store:}
            \State $\batches = \{\}$ \Comment{Set of tuples $(m, \sn, q)$ of receives messages}
            \State $\poas = \{\}$ \tab\Comment{Set of proofs of availability}
            \State $\myBatches = \text{map }[\sn] \to \batch$
            \State $\poaVotes = \text{map }[\sn] \to \signature$
        \EndMultiline

        \algspace
        \Multiline \textbf{timers:}
            \State $\tQCVote$ timer \tab\tab\Comment{Used to delay the QC-vote when not all data is available}
            \State $\tRoundTimeout$ timer \Comment{Used to time out faulty leaders}
        \EndMultiline

        \algspace
        \Multiline \textbf{notation and parameters:}
            \State $p$ -- this replica
            \State $\leader_r$ -- the leader of round $r$
            \State $\nSubBlocks$ -- the number of sub-blocks (groups) of optimistically proposed batch hashes per block (see \Cref{sec:system:signature})
            \State $\storageRequirement$ -- the availability requirement (see \Cref{sec:raptr:intuition})
            \State $\Delta$ -- the upper bound on message delivery time in after GST (see \Cref{sec:prelim:model})
            \State $\epsilon\Delta$ -- the duration of the small delay before QC-voting when not all data is available
            \State $\chain(\qc)$ -- the sequence of tuples $(B_i, \prefix_i)$ representing the chain ending in $\qc$\atadd{ (see \Cref{sec:notation})}
            \State $\messages(\chain(\qc))$ -- the sequence of messages in $\chain(\qc)$\atadd{ (see \Cref{sec:notation})}
            \State \atadd{$\rank(\qc)$ -- tuple $(\qc.\round, \qc.\prefix)$ (see \Cref{sec:notation})}
        \EndMultiline

        \algspace
        \Multiline \textbf{multisignature interface:}
            \State $\sk_p$ -- the secret key of this \atreplace{node}{replica} $p$
            \State $\pk_j$ -- the public key of \atreplace{node}{replica} $j$
            \State $\psign, \pver, \comb, \ver$ -- signing, combining, and verifying digital signatures (see \Cref{sec:prelim:crypto}).
        \EndMultiline



    \end{smartalgorithmic}

    \addtocounter{raptrparts}{1}
    \caption{{\Raptr}: part {\theraptrparts} (definitions)}
    \label{alg:raptr-first}
    \label{alg:raptr:definitions}
\end{algorithm*}

\begin{algorithm*}[p]
    \begin{smartalgorithmic}[1]

        \Multiline \textbf{structures:}
            \State Quorum Certificate (QC): $(\round, \hash, \votePrefixes, \signature)$, where $\votePrefixes$ is a set of tuples $(\replicaId, \prefix)$
                \Statex\tab $\qc.\prefix := \storageRequirement$'th $\max\{\,\prefix \mid ({\replicaId}, \prefix) \in \qc.\votePrefixes\,\}$ \Comment{i.e., disregard $S-1$ largest prefixes and take maximum}

            \State Commit Certificate (CC): $(\round, \hash, \votePrefixes, \signature)$, where $\votePrefixes$ is a set of tuples $(\replicaId, \prefix)$
                \Statex\tab $\cc.\commitPrefix := \min \{\,\prefix \mid ({\replicaId}, \prefix) \in \cc.\votePrefixes \,\}$
                \Statex\tab $\cc.\extendPrefix := \max \{\,\prefix \mid ({\replicaId}, \prefix) \in \cc.\votePrefixes\,\}$

            \State Timeout Certificate (TC): $(\round, \voteData, \signature)$, where $\voteData$ is a set of tuples $(\replicaId, (\round, \prefix))$
                \Statex\tab $\tc.\extendRank := \max \{\,\rank \mid ({\replicaId}, \rank) \in \tc.\voteData)\,\}$ 

            \State Entry Reason: one of $\rFullQC(\atadd{\qc})$, $\rCC(\cc\atadd{, \qc})$, or $\rTC(\tc\atadd{, \qc})$
            \Statex\tab \atadd{$\reason.\round := \textbf{if } \reason = \rTC(\tc, \qc) \textbf{ then } \tc.\round + 1 \textbf{ else } \reason.\qc.\round + 1$}
            \State Block: $(\round\atremove{, \qcParent},  \entryReason, \payload)$
            \Statex\tab \atadd{$\block.\qcParent := \block.\entryReason.\qc$}
            \State \atadd{Proof of Availability (PoA): $(\hash, \sn, \replicaId, \votes, \signature)$, where $\votes$ is a set of replica IDs}
        \EndMultiline





        \algspace
        \LineComment{\atadd{If any of the verification functions fails, the message that contained the invalid certificate is ignored.}}
        \LineComment{\atadd{In addition to these verification functions, all partial signatures in all messages are implicitly assumed to be verified.}}

        \algspace[0.5]
        \LineComment{\atadd{{\Raptr} verification functions}}
        
        \algspace[0.5]
        \Function{$\VerifyQC$}{$\qc$}
        \Comment{implicitly invoked each time a QC is received in a message\atremove{, the message is ignored if the check fails}}
        \label{line:verify-qc:start}
            \BreakableLine[\Return{}] $|\qc.\votePrefixes| \ge \quorumSize \tand$
                \RaggedBreak $\ver(\{\, (\pk_{q}, \mQCVote||\qc.\hash||\qc.\round||\prefix) \mid (q, \prefix) \in \qc.\votePrefixes \,\}, \qc.\signature)$
            \EndBreakableLine
        \EndFunction

        \algspace[0.5]
        \Function{$\VerifyCC$}{$\cc$} \Comment{implicitly invoked each time a CC is received\atadd{ in a message}} \label{line:verify-cc:start}
            \BreakableLine[\Return{}] $|\cc.\votePrefixes| = \quorumSize \tand$
                \RaggedBreak $\ver(\{\,(\pk_q, \mCCVote||\cc.\hash||\cc.\round||\prefix) \mid (q, \prefix) \in \cc.\votePrefixes \,\}, \cc.\signature)$
            \EndBreakableLine
            
        \EndFunction

        \algspace[0.5]
        \Function{$\VerifyTC$}{$\tc$} \Comment{implicitly invoked each time a TC is received\atadd{ in a message}}
            \BreakableLine[\Return{}] $|\tc.\voteData| \ge \quorumSize \tand$
                \RaggedBreak $\ver(\{\, (\pk_q, \mTCVote||\tc.\round||\voteRound||\votePrefix) \mid (q, (\voteRound, \votePrefix)) \in \tc.\voteData \,\}, \tc.\signature)$
            \EndBreakableLine
        \EndFunction
        
        \algspace[0.5]
        \Function{$\VerifyEntryReason$}{$\atremove{r, \qc, }\reason$} 
        \Comment{\atreplace{invoked explicitly on \cref{line:advance-round:if,line:recv-proposal:if}}{implicitly invoked each time an Entry Reason is received in a message}} \label{line:verify-entry-reason:start}
            \If {$\reason = \rFullQC\atadd{(\qc)}$} \label{line:verify-entry-reason:full-qc-if}
                \State \Return{$\qc.\prefix = \nSubBlocks$} \label{line:verify-entry-reason:full-qc-return}
            \ElsIf{$\reason = \rCC(\cc\atadd{, \qc})$} \label{line:verify-entry-reason:cc-if}
                \State \Return{$\cc.\round = \qc.\round\atremove{ = r - 1} \tand \qc.\prefix \ge \cc.\extendPrefix$} \label{line:verify-entry-reason:cc-return}
            \ElsIf{$\reason = \rTC(\tc\atadd{, \qc})$} \label{line:verify-entry-reason:tc-if}
                \State \Return{$\atremove{\tc.\round = r - 1 \tand }\rank(\qc) \ge \tc.\extendRank\atadd{ \tand \qc.\round \le \tc.\round}$} \label{line:verify-entry-reason:tc-return}
            \EndIf
        \EndFunction \label{line:verify-entry-reason:end}

        \algspace[0.5]
        \LineComment{\atadd{Quorum store verification functions:}}

        \algspace[0.5]
        \Function{$\VerifyPoA$}{$\poa$} \Comment{implicitly invoked each time a PoA is received in a message}
            \State \Return{$|\poa.\votes| \ge \quorumSize \tand \ver(\{(\pk_q, \mPoAVote||\poa.\hash||\poa.\sn||\poa.\replicaId) \mid q \in votes \}, \poa.\signature)$}
        \EndFunction
    \end{smartalgorithmic}

    \addtocounter{raptrparts}{1}
    \caption{{\Raptr}: part {\theraptrparts} (structures and verification functions)}
    \label{alg:raptr:structures-and-verification}
\end{algorithm*}

\begin{algorithm*}[p]
    \begin{smartalgorithmic}[1]

        \Function{$\aBcast$}{$m, \sn$} \Comment{Multicast the batch of transactions to all replicas} 
            \State $\myBatches[\sn] = m$
            \State \Multicast{$\mBatch, m, \sn$}
        \EndFunction

        \algspace
        \UponReceiving{$\mBatch, m,\sn$}{replica $q$}
            \State $\batches.\Add((m, \sn, q))$
            \State \Send{$\mPoAVote, \sn, \psign(\sk_p, \mPoAVote||H(m)||\sn||q)$}{$q$}
        \EndUpon

        \algspace
        \UponReceiving{$\mPoAVote, \sn, \sigma$}{replica $q$}
            \State let $m = \myBatches[\sn]$
            \State $\poaVotes[\sn][q] = \sigma$
            \If{$|\poaVotes[\sn]| = \quorumSize$}
                \State let $\votes = \{\{(q, \sigma) \mid q \in \Pi \tand \sigma = \poaVotes[q]\})\}$
                \State \Multicast{$\mPoA, (H(m), \sn, p, \{q \mid (q, \sigma) \in \votes\}, \comb(\votes))$}
            \EndIf
        \EndUpon

        \algspace
        \UponReceiving{$\mPoA, \poa$}{any replica}
            \State $\poas.\Add(\poa)$
            \State let $(h, \sn, q, \votes, \sigma) = \poa$
            \If {$\batches$ does not a batch with hash $h$}
                \State fetch the batch from the replicas in $\votes$
            \EndIf
        \EndUpon

        \algspace
        \Function{$\AvailablePrefix$}{$B$}
            \State \Return the length of the longest available prefix of batches referenced in $B$ (a number between 0 and $\nSubBlocks$)
        \EndFunction

        \algspace
        \Function{$\GetPayload$}{$\qc$}
            \State \Return a block that consists of: 
            \State \tab(1): all PoAs in $\poas$ that are not included in $\chain(\qc)$
            \State \tab(2): all batches in $\batches$ that are not included in $\chain(\qc)$ and for which there are no PoAs in $\poas$,
            \State \tab\tab split into $\nSubBlocks$ groups (called \emph{sub-blocks})
        \EndFunction

        \algspace
        \Function{$\OnNewBlock$}{$B$} \Comment{implicitly invoked whenever a block is received}
            \State add the PoAs in the block to $\poas$ and fetch the data
            \State try fetching the referenced batches from their creators and the leader proposing the block
        \EndFunction

        \algspace
        \Function{$\FetchQCData$}{$\qc$} \Comment{invoked in $\OnNewQC$}
            \State recursively fetch all data in $\chain(\qc)$, if missing locally
        \EndFunction
    \end{smartalgorithmic}

    \addtocounter{raptrparts}{1}
    \caption{\atadd{{\Raptr}: part {\theraptrparts} (quorum store)}}
    \label{alg:raptr:quorum-store}
\end{algorithm*}

\begin{algorithm*}[p]
    \begin{smartalgorithmic}[1]
        \LineComment*{At the beginning of the protocol, all \atreplace{nodes}{replicas} enter round 1.}
        \UponStart \label{line:start}
            \State $\TryAdvanceRound(1, \rFullQC\atadd{(\qcGenesis)})$
        \EndUpon

        \algspace
        \LineComment*{\textbf{\TryAdvanceRound:} Upon entering round $r$, notify the leader so that it can also advance to this round,}
        \LineComment*{stop the $\tQCVote$ timer, and reset the $\tRoundTimeout$ timer to $\roundTimeoutDuration$.}
        \LineComment*{If this \atreplace{node}{replica} is the leader, create a block and multicast it to all participants.}
        \Function{\TryAdvanceRound}{$r, \reason$}
            \If{$\atadd{\reason.\round = r \tand }r > \rCur$}
                \State $\rCur = r$
                \State $\entryReason = \reason$
                \If {$p = \leader_r$}
                    \State let $\payload = \GetPayload(\atadd{\reason.\qc})$
                    \State let $B = (\rCur\atremove{, \qcHigh}, \reason, \payload)$
                    \State \Multicast{$\mPropose, B$}
                \Else
                    \State \Send{$\mAdvanceRound, r\atremove{, \qcHigh}, \reason$}{$\leader_r$} 
                \EndIf

                \State stop the $\tQCVote$ timer if running
                \State reset the $\tRoundTimeout$ timer to $\roundTimeoutDuration$
            \EndIf
        \EndFunction

        \algspace
        \LineComment*{Upon receiving a valid $\mAdvanceRound$ message, execute $\OnNewQC$ and $\TryAdvanceRound$.}
        \UponReceiving{$\mAdvanceRound\atremove{, r, \qc}, \reason$}{any node} \label{line:advance-round:start} 
                \State $\OnNewQC(\atadd{\reason.}\qc)$
                \State $\TryAdvanceRound(\atreplace{r}{\reason.\round}, \reason)$
        \EndUpon

        \algspace
        \LineComment*{\textbf{\OnNewQC:} Each time when receiving or forming a new QC:}
        \LineComment*{\quad 1. Update $\qcHigh$.}
        \LineComment*{\quad 2. Issue a CC-vote if not yet CC-voted or TC-voted in this round.}
        \LineComment*{\quad 3. Upon observing a full-prefix QC for round $r$, advance to round $r+1$.}
        \LineComment*{\textbf{Two-chain commit\atadd{ (optional)}:} Whenever there exists a certified block $B$ with the parent from the adjacent round, i.e.,}
        \LineComment*{$B.\round = B.\qcParent.\round + 1$, commit $B.\qcParent$}

        \Function{\OnNewQC}{$\qc$}
            \If {$\rank(\qc) > \rank(\qcHigh)$}
                $\qcHigh = \qc$ \label{line:on-new-qc:update-qc-high}
            \EndIf
            \If {$\tnot {\ccVoted[\qc.\round]} \tand \qc.\round > \rTimeout$} \label{line:on-new-qc:if-cc-vote}
                \State {$\ccVoted[\qc.\round] = \true$}
                \State \Multicast{$\mCCVote, \qc, \psign(\sk_p, \mCCVote||\qc.\hash||\qc.\round||\qc.\prefix)$} \label{line:on-new-qc:cc-vote}
            \EndIf
            \If{$\qc.\prefix = \nSubBlocks$}
                  $\TryAdvanceRound(\qc.\round + 1, \rFullQC\atadd{(\qc)})$
            \EndIf

            \State $\FetchQCData(\qc)$

            \Upon{all blocks in $\chain(qc)$ are locally available}
                \If{$\exists B \in \chain(\qc) \text{ such that } B.\round = B.\qcParent.\round + 1 
            \tand \rank(B.\qcParent) > \rank(\qcCommitted)$}
                    \State $\CommitQC(B.\qcParent)$ \label{line:on-new-qc:commit-qc} \Comment{two-chain commit rule\atadd{ \textbf{(optional)}}}
                \EndIf
            \EndUpon
        \EndFunction

        \algspace
        \LineComment*{\textbf{Deliver:} upon committing a QC, deliver all new transactions in its chain to the application layer.}
        \Function{$\CommitQC$}{$\qc$} \label{line:commit-qc:start}
            \State $\qcCommitted = \qc$
            \Upon {all data in $\chain(\qc)$ is available locally}
                \For {$(m, \sn, q) \in \messages(\chain(\qc))$}
                    \If {$(m, \sn, q)$ has not yet been delivered} \label{line:commit-qc:dedup-if}
                        \State $\aDeliver(m, \sn, q)$ \label{line:commit-qc:a-deliver}
                    \EndIf
                \EndFor
            \EndUpon
        \EndFunction
    \end{smartalgorithmic}

    \addtocounter{raptrparts}{1}
    \caption{{\Raptr}: part {\theraptrparts} (key auxiliary functions)}
    \label{alg:raptr:aux}
\end{algorithm*}

\begin{algorithm*}[p]
    \begin{smartalgorithmic}[1]
        \LineComment*{Upon receiving a valid block $B = (r\atremove{, \qcParent}, \reason, \payload)$ from $L_r$ for the first time, if $r \ge \rCur$ and $r > \rTimeout$,}
        \LineComment*{store the proposed block, execute $\OnNewQC$ and $\TryAdvanceRound$, and start the $\tQCVote$ timer.}
        \UponReceiving{\mPropose, $B$}{$q$}
            \State let $(r\atremove{, \qcParent}, \reason, \payload) = B$
            \State $\OnNewQC(\atadd{\reason.\qc})$
            \If {$q = \leader_r \tand r \ge \rCur \tand r > \rTimeout \tand \proposal[r] = \bot \tand \atreplace{\VerifyEntryReason(r, \qcParent, \reason)}{\reason.\round = r}$} \label{line:recv-proposal:if}
                \State $\proposal[r] = B$ \label{line:recv-proposal:set-proposal-var}
                \State $\TryAdvanceRound(r, \reason)$
                \State start the $\tQCVote$ timer set to $\qcVoteTimerDuration$ \label{line:recv-proposal:start-timer}
            \EndIf
        \EndUpon

        \algspace
        \LineComment*{A \atreplace{node}{replica} issues a QC-vote in its current round up to 2 times:}
        \LineComment*{\quad 1. Once a fixed timer has elapsed since receiving the block.}
        \LineComment*{\quad 2. When the \atreplace{node}{replica} can issue a full-prefix QC-vote.}
        \LineComment*{The \atreplace{node}{replica} QC-votes only if $\rCur > \rTimeout$.}
        %
        %
        \Upon{$\tQCVote$ timer expires}
            \State $\QCVote()$ \label{line:qc-vote-call-timer}
        \EndUpon

        \algspace[0.5]
        \Upon{$\AvailablePrefix(\proposal[\rCur]) = \nSubBlocks$}
            \State $\QCVote()$ \label{line:qc-vote-call-full}
        \EndUpon

        \algspace[0.5]
        \Function{$\QCVote$}{}
            \State let $\prefix = \AvailablePrefix(\proposal[\rCur])$
            \If {$\rCur > \rTimeout \tand \lastQCVote < (\rCur, \prefix)$} \label{line:qc-vote:if}
                \State $\lastQCVote = (\rCur, \prefix)$
                \State \Multicast{$\mQCVote, \rCur, \prefix, H(\proposal[\rCur]), \psign(\sk_p, \mQCVote||H(\proposal[\rCur])||\rCur||\prefix)$} \label{line:qc-vote:sign-and-send}
            \EndIf
        \EndFunction


        \algspace
        \LineComment*{{A \atreplace{node}{replica} only processes QC-votes in round $\rCur$ and higher.}}
        \LineComment*{It can form a QC up to 2 times{ in a round}, when it has received a quorum of QC-votes with the same block digest and either:}
        \LineComment*{\quad 1. the \atreplace{node}{replica} has not yet formed or received any QC in this round; or}
        \LineComment*{\quad 2. the \atreplace{node}{replica} can form the full-prefix QC, i.e., at least $\storageRequirement$ of the votes have prefix $\nSubBlocks$.}
        \LineComment*{Upon forming a QC, execute $\OnNewQC$.}
        \UponReceiving{$\mQCVote, r, \prefix, h, \sigma$}{\atreplace{node}{replica} $q$} \label{line:recv-qc-vote:start}
            \If {{$r \ge \rCur \tand (\qcVotes[r][h][q] = \bot \tor \text{the prefix in }\qcVotes[r][h][q] < \prefix)$}}
                \State $\qcVotes[r][h][q] = (\prefix, \sigma)$ 
                
                \If {$|\qcVotes[r][h]| \ge \quorumSize \tand (\qcHigh.\round < r \tor |\{\, q \mid \atadd{\text{ the prefix in }}\qcVotes[r][h][q] = \nSubBlocks \,\}| \ge \storageRequirement)$}
                    \State let $\votes = \{\,(q, \prefix, \sigma) \mid q \in \Pi \tand (\prefix, \sigma) = \qcVotes[r][h][q] \,\}$
    
                    \State let $\votePrefixes = \{\,(q, \prefix) \mid (q, \prefix, \sigma) \in \votes\,\}$ 
                    \State let $\qc = (r, h, \votePrefixes, \comb(\{\,(q,\sigma) \mid (q, \prefix, \sigma) \in \votes \,\}$ \label{line:recv-qc-vote:form-qc}
                    \State $\OnNewQC(\qc)$                
                \EndIf
            \EndIf
        \EndUpon
    \end{smartalgorithmic}

    \addtocounter{raptrparts}{1}
    \caption{{\Raptr}: part {\theraptrparts} (QC-votes)}
    \label{alg:raptr:qc-votes}
\end{algorithm*}

\begin{algorithm*}[p]
    \begin{smartalgorithmic}[1]
        \LineComment*{Upon receiving a CC-vote, execute $\OnNewQC$.}
        \LineComment*{Upon receiving a quorum of CC-votes, form a CC, commit the $\quorumSize$'th maximum QC from the votes,} \LineComment*{and advance to the next round.}
        \LineComment*{A \atreplace{node}{replica} can form multiple CCs in the same round, each time increasing the committed prefix.}
        \UponReceiving{$\mCCVote, \qc, \sigma$}{\atreplace{node}{replica} $q$} \label{line:recv-cc-vote}
            \State $\OnNewQC(\qc)$
            \State let $r = \qc.\round$
            \State $\ccVotes[r][q] = (\qc, \sigma)$
            \If {$|\ccVotes[r]| \ge \quorumSize$} \label{line:recv-cc-vote:if}
                \State let $\votes = \{\, \text{set of $\quorumSize$ votes from $\ccVotes[r]$ with the highest QCs in format $(\sender, \qc, \signature)$} \,\}$ \label{line:recv-cc-vote:select-votes}
                \State let $\qcMin =  \min \{\, \qc \mid (q, \qc, \sigma) \in \votes \,\}$ \label{line:recv-cc-vote:qc-min}
                \State \atadd{let $\qcMax = \max \{ \qc \mid (q, \qc, \sigma) \in \votes) \}$}
                \If {$\rank(\qcMin) > \rank(\qcCommitted)$}
                    \State let $\votePrefixes = \{(q, \qc.\prefix) \mid (q, \qc, \sigma) \in \votes \}$ \Comment{necessary for signature verification}
                    \State let $\hash = \qc.\hash$ \Comment{This could be any QC from $\votes$. All same-round QCs have the same block hash.} 
                    \State let $\cc = (r, \hash, \votePrefixes, \comb(\{\, (q, \sigma) \mid (q, \qc, \sigma) \in \votes \,\}))$ \label{line:recv-cc-vote:form-cc}
                    \State $\CommitQC(\qcMin)$ \label{line:recv-cc-vote:commit-qc}
                    \State $\TryAdvanceRound(r+1, \rCC(cc\atadd{, \qcMax}))$
                \EndIf
            \EndIf
        \EndUpon

        \algspace
        \LineComment*{When the timeout expires, multicast a signed {TC-vote} with $\qcHigh$ attached.}
        \Upon{the $\tRoundTimeout$ timer expires} \label{line:round-timeout:start}
            \State $\rTimeout = \rCur$
            \State \Multicast{$\mTCVote, \rCur\atadd{, \entryReason}, \qcHigh, \psign(\sk_p, \mTCVote||\rCur||\qcHigh.\round||\qcHigh.\prefix)$} \label{line:round-timeout:tc-vote}
        \EndUpon

        \algspace
        \LineComment*{A \atreplace{node}{replica} only processes TC-votes in round $\rCur$ and higher.}
        \LineComment*{Upon receiving a valid {TC-vote}, execute $\OnNewQC$.}
        \LineComment*{Upon gathering a quorum of matching TC-votes, form the TC and execute $\TryAdvanceRound$.}
        \UponReceiving{$\mTCVote, r\atadd{, \entryReason}, \qc, \sigma$}{\atreplace{node}{replica} $q$}
            \If {\atadd{$\reason.\round = r$}}
                \State $\OnNewQC(\qc)$
                \State \atadd{$\TryAdvanceRound(r, \entryReason)$} \label{line:recv-tc-vote:catch-up-round}
                \State $\tcVotes[r][q] = (\qc, \sigma)$
                \If {$|\tcVotes[r]| = \quorumSize$}
                    \State let $\votes = \{\,(q, \qc, \sigma) \mid q \in \Pi \tand (\qc, \sigma) = \tcVotes[r][q] \,\}$
                    
                    \State let $\voteData = \{\, (q, \rank(\qc)) \mid (q, \qc, \sigma) \in \votes \,\}$ 
                    \State \atadd{let $\qcMax = \max \{\qc \mid (q, \qc, \sigma) \in \votes\}$}
                    \State let $\tc = (r, \voteData, \comb(\{\, (q, \sigma) \mid (q, \qc, \sigma) \in \tcVotes[r] \,\}))$ \label{line:recv-tc-vote:form-tc}
                    \State $\TryAdvanceRound(r + 1, \rTC(\tc\atadd{, \qcMax}))$ \label{line:recv-tc-vote:advance-round}
                \EndIf
            \EndIf
        \EndUpon
    \end{smartalgorithmic}

    \addtocounter{raptrparts}{1}
    \caption{{\Raptr}: part {\theraptrparts} (commit and timeouts)}
    \label{alg:raptr-last}
    \label{alg:raptr:cc-and-tc}
\end{algorithm*}

\FloatBarrier
\section{Formal proof of correctness}\label{sec:proof}

    


\subsection{Safety proof}\label{subsec:proof:safety}

\begin{lemma}[Supermajority Quorum Intersection]
    For any two sets of replicas $Q_1, Q_2$ such that $|Q_1| \ge \quorumSize$ and $|Q_2| \ge \quorumSize$, there is at least one honest replica in $Q_1 \cap Q_2$.
\end{lemma}
\begin{proof}
    $|Q_1 \cap Q_2| = |Q_1| + |Q_2| - |Q_1 \cup Q_2| \ge 2 \quorumSize - n \ge f + 1$.
    Since there are at least $f+1$ replicas in the intersection and at most $f$ of them are malicious, at least one replica in the intersection must be honest.
\end{proof}

\begin{lemma} \label{lem:no-conflicting-qcs-same-round}
    There cannot be two valid QCs in the same round for two different blocks, i.e., $\qc_1.\round = \qc_2.\round$, then $\qc_1.\hash = \qc_2.\hash$.
\end{lemma}
\begin{proof}
    Let $\qc_1$ and $\qc_2$ be two valid QCs such that $\qc_1.\round = \qc_2.\round$.
    By definition (\Cref{sec:raptr:description}, \cref{line:verify-qc:start}), a valid QC contains an aggregate signature from a quorum of replicas.
    From the unforgeability of aggregate signatures and the super-majority quorum intersection property, a signature from at least one honest replica contributed to both $\qc_1$ and $\qc_2$.
    %
    %
    Honest replicas only sign QC-votes on the first proposal they receive in a round from the leader (\cref{line:qc-vote:sign-and-send,line:recv-proposal:set-proposal-var,line:recv-proposal:if}).
    Thus, an honest replica will never sign two QC-votes with different block hashes and the same round number.
    Therefore, $\qc_1.\hash = \qc_2.\hash$. 
\end{proof}


\begin{lemma} \label{lem:finality-main}
    If $\qc_1$ is a directly committed QC (\cref{line:commit-qc:start}), then for any $\qc_2$ such that $\rank(\qc_1) \le \rank(\qc_2)$, $\qc_1 \preceq \qc_2$.
\end{lemma}
\begin{proof}
    Let us use induction on $\qc_2.\round$.
    
    \textbf{Base case:}
    The case when $\qc_2.\round = \qc_1.\round$ trivially follows from \Cref{lem:no-conflicting-qcs-same-round}.

    \textbf{Induction hypothesis:}
    Suppose that, for some $r \ge \qc_1.\round$, for any $\qc_2$ such that $\qc_2.\round \le r - 1$ and $\rank(\qc_1) \le \rank(\qc_2)$, it holds that $\qc_1 \preceq \qc_2$.
    %
    
    \textbf{Induction step:}
    Let us now prove that for $\qc_2.\round = r$.
    Consider $\qc^* = \block(\qc_2).\qcParent$.
    Since a valid block must have a parent from a lower round (\crefrange{line:verify-entry-reason:start}{line:verify-entry-reason:end}), $\qc^*.\round \le \qc_2.\round - 1$.
    If $\rank(\qc^*) \ge \rank(\qc_1)$, by the induction hypothesis, $\qc_1 \preceq \qc^* \preceq \qc_2$.
    The rest of the proof will be dedicated to showing that, indeed, $\rank(\qc^*) \ge \rank(\qc_1)$.
    %
    
    There are 2 ways that a QC can get directly committed: when a \atreplace{node}{replica} forms a CC (\cref{line:recv-cc-vote:commit-qc}) and through the 2-chain commit rule (\cref{line:on-new-qc:commit-qc}).
    {We will consider the two cases separately, in \Cref{lem:finality-cc,lem:finality-2chain}.}
\end{proof}

\begin{lemma} \label{lem:finality-cc}
    Let $\qc_1$ be a QC directly committed through a CC (\cref{line:recv-cc-vote:commit-qc}), $\qc_2$ be a valid QC such that $\qc_2.\round > \qc_1.\round$, and $\qc^* := \block(\qc_2).\qcParent$. Then, $\rank(\qc^*) \ge \rank(\qc_1)$.
\end{lemma}
\begin{proof}
    Let $\cc_1$ be the CC formed on \cref{line:recv-cc-vote:form-cc} right before committing $\qc_1$ on \cref{line:recv-cc-vote:commit-qc}, and let $Q_1$ be the set of replicas whose CC-votes were used to form it.
    %

    Let us consider $\block(\qc_2).\entryReason$. It must be valid according to the $\VerifyEntryReason$ function (\cref{line:verify-entry-reason:start}). Otherwise, no honest replica would have voted for $\qc_2$ and it could not have been formed.
    %
    There are 3 cases to consider:
    \begin{enumerate}
        \item $\block(\qc_2).\entryReason = \rFullQC\atadd{(\qc^*)}$ (\cref{line:verify-entry-reason:full-qc-if}): By \cref{line:verify-entry-reason:full-qc-return}, $\qc^*.\round = \qc_2.\round - 1 \ge \qc_1.\round$ and $\qc^*.\prefix = \nSubBlocks \ge \qc_1.\prefix$ $\Rightarrow$ $\rank(\qc^*) \ge \rank(\qc_1)$ (QED). \label{itm:proof:cc-commit:full-qc}

        \item $\block(\qc_2).\entryReason = \rCC(\cc_2\atadd{, \qc^*})$ (\cref{line:verify-entry-reason:cc-if}): 
        %
        If $\cc_2.\round > \qc_1.\round$, then, by \cref{line:verify-entry-reason:cc-if}, $\qc^*.\round = \cc_2.\round > \qc_1.\round$ (QED).
        Since $\cc_2.\round = \qc_2.\round - 1 \ge \qc_1.\round$, the only case left to consider is when $\cc_2.\round = \qc_1.\round$.
        In this case, $\qc^*.\round = \cc_2.\round = \qc_1.\round$ and we have left to prove that $\qc^*.\prefix \ge \qc_1.\prefix$.

        Let $Q_2$ be the set of replicas whose votes were used to form $\cc_2$ (i.e., whose votes are recorded in $\cc_2.\votePrefixes$).
        By the supermajority quorum intersection, there must be an honest replica in the intersection of $Q_1$ and $Q_2$.
        Since honest replicas only issue one CC-vote in any given round (\cref{line:on-new-qc:if-cc-vote}), by the unforgeability property of aggregate signatures, there must be at least one CC-vote used in both $\cc_1$ and $\cc_2$. Let $\qc'$ be the $\qc$ attached to that vote.
        %
        %
        Since we take the minimum when choosing which QC to commit (\cref{line:recv-cc-vote:qc-min}) and require $\qc^*$ to have the prefix at least as large as the maximum of the votes used to form $\cc_2$ (\cref{line:verify-entry-reason:cc-return}), 
        $\qc^*.\prefix \ge \max \{\, \prefix \mid (q, \prefix) \in \cc_2.\votePrefixes \,\} \ge \qc'.\prefix \ge \min \{\, \prefix \mid (q, \prefix) \in \cc_1.\votePrefixes \,\} = \qc_1.\prefix$.
        Hence, $\rank(\qc^*) \ge \rank(\qc_1)$ (QED).

        \item $\block(\qc_2\atadd{, \qc^*}).\entryReason = \rTC(\tc)$ (\cref{line:verify-entry-reason:tc-if}): 
        %
        %
        Let $Q_2$ be the quorum of replicas that TC-voted for $\tc$ (i.e., whose votes are recorded in $\tc.\voteData$).
        By the unforgeability property of aggregate signatures, for each honest replica in $Q_2$, it must be that each honest replica in $Q_2$ issues $\mTCVote$ message (\cref{line:round-timeout:tc-vote}).
        By the supermajority quorum intersection, there must be an honest replica $q \in (Q_1 \cap Q_2)$.
        Let $\qc_{c}$ and $\qc_{t}$ be the QCs that $q$ attached to its CC-vote (\cref{line:on-new-qc:cc-vote}) and TC-vote (\cref{line:round-timeout:tc-vote}) respectively.
        Honest replicas only CC-vote if they have not yet TC-voted in this or higher rounds (\cref{line:on-new-qc:if-cc-vote}) and 
        $\tc.\round = \qc_2.\round - 1 \ge \qc_1.\round = \cc_1.\round$.
        Hence, $q$ must have CC-voted for $\cc_1$ before TC-voting for $\tc$ and $\rank(\qc_{t}) \ge \rank(\qc_{c})$.
        Since we take the minimum when choosing which QC to commit (\cref{line:recv-cc-vote:qc-min}) and verify that $\rank(\qc^*) \ge \tc.\extendRank = \max \{\,\rank \mid (q,\rank) \in \tc.\voteData) \,\}$ (\cref{line:verify-entry-reason:tc-return}), $\rank(\qc^*) \ge \rank(\qc_{t}) \ge \rank(\qc_{c}) \ge \rank(\qc_1)$ (QED).
        %
    \end{enumerate}
\end{proof}

\begin{lemma} \label{lem:finality-2chain}
    Let $\qc_1$ be a QC directly committed via the two-chain rule (\cref{line:on-new-qc:commit-qc}), $\qc_2$ be a valid QC such that $\qc_2.\round > \qc_1.\round$, and $\qc^* := \block(\qc_2).\qcParent$. Then, $\rank(\qc^*) \ge \rank(\qc_1)$.
\end{lemma}
\begin{proof}
    Let $\widetilde\qc$ denote the child of $\qc_1$ in the next round (i.e., $\block(\widetilde\qc).\qcParent = \qc_1$ and $\widetilde\qc.\round = \qc_1.\round + 1$) that triggered the commit on \cref{line:on-new-qc:commit-qc}.
    %
    %
    %
    As in \Cref{lem:finality-cc}, let us consider all 3 possible options for $\block(\qc_2).\entryReason$:
    \begin{enumerate}
        \item $\block(\qc_2).\entryReason = \rFullQC\atadd{(\qc^*)}$ (\cref{line:verify-entry-reason:full-qc-if}): Identical to the case~(\ref{itm:proof:cc-commit:full-qc}) in \Cref{lem:finality-cc}. By \cref{line:verify-entry-reason:full-qc-return}, $\qc^*.\round = \qc_2.\round - 1 \ge \qc_1.\round$ and $\qc^*.\prefix = \nSubBlocks \ge \qc_1.\prefix$ $\Rightarrow$ $\rank(\qc^*) \ge \rank(\qc_1)$ (QED).

        \item $\block(\qc_2).\entryReason = \rCC(\cc\atadd{, \qc^*})$ (\cref{line:verify-entry-reason:cc-if}): If $\qc_2.\round = \qc_1.\round + 1$, then, by \Cref{lem:no-conflicting-qcs-same-round}, $\qc^* = \block(\qc_2).\qcParent = \block(\widetilde\qc).\qcParent = \qc_1$ (QED).
        Otherwise ($\qc_2.\round > \qc_1.\round + 1$), by \cref{line:verify-entry-reason:cc-return}, $\qc^*.\round = \qc_2.\round - 1 > \qc_1.\round$ (QED).

        \item $\block(\qc_2).\entryReason = \rTC(\tc\atadd{, \qc^*})$ (\cref{line:verify-entry-reason:tc-if}):
        Analogously to the previous case, if $\qc_2.\round = \qc_1.\round + 1$, then $\qc^* = \qc_1$ (QED).
        Hence, let us consider the case when $\qc_2.\round > \qc_1.\round + 1$.
        let $Q_1$ be the set of replicas that QC-voted for $\widetilde\qc$, and let $Q_2$ the set of replicas that TC-voted for $\tc$ (i.e., whose votes are recorded in $\tc.\voteData$).
        %
        There must be an honest replica $q \in (Q_1 \cap Q_2)$.
        Let $\qc_{t}$ be the QC that $q$ attached to its TC-vote (\cref{line:round-timeout:tc-vote}).
        Honest replicas only QC-vote if they have not yet TC-voted in this or higher rounds (\cref{line:qc-vote:if}).
        Hence, as $\qc_2.\round > \qc_1.\round + 1 = \widetilde\qc.\round$, $q$ must have QC-voted for $\widetilde\qc$ before TC-voting for $\tc$ and $\rank(\qc_1) \le \rank(\qc_{t}) \le \tc.\extendRank \le \rank(\qc^*)$ (QED).
    \end{enumerate}
    
\end{proof}

\begin{lemma} \label{lem:all-committed-qcs-are-compatible}
    For any two directly committed QCs, $\qc_1$ and $\qc_2$, either $\qc_1 \preceq \qc_2$ or vice versa.
\end{lemma}
\begin{proof}
    Follows directly from \Cref{lem:finality-main} and the fact that \atreplace{IDs}{the ranks} are totally ordered. 
\end{proof}

\begin{theorem}
    {\Raptr} satisfies the total order \atreplace{property}{and non-duplication properties} of BAB (\Cref{def:bab}).
\end{theorem}
\begin{proof}
    For the sake of contradiction, suppose that an honest replica $p$ outputs $\aDeliver(m, \sn, q)$ before $\aDeliver(m', \sn', q')$ and an honest replica $p'$ outputs $\aDeliver(m', \sn', q')$ before $\aDeliver(m, \sn, q)$.
    The only way honest replicas deliver messages is on \cref{line:commit-qc:a-deliver}, when directly committing some QC (and indirectly committing its ancestors).
    \atrev{Hence, there must be two directly committed QCs, $\qc$ and $\qc'$, such that $(m, \sn, q)$ appears in $\messages(\chain(\qc))$ 
    and $(m', \sn', q')$ either appears in $\messages(\chain(\qc))$ after $(m, \sn, q)$ or not at all.
    Conversely for $\qc'$.
    However, by \Cref{lem:all-committed-qcs-are-compatible}, either $\messages(\chain(\qc_1))$ is a prefix of $\messages(\chain(\qc_2))$ or vice versa. A contradiction.}

    \atadd{Non-duplication is trivially enforced when delivering the messages (\cref{line:commit-qc:dedup-if}).}
\end{proof}

\subsection{Liveness proof}\label{subsec:proof:liveness}

\begin{lemma} \label{lem:correct-leaders-commit}
    If $\leader_r$ is correct and the first time a correct \atreplace{node}{replica} enters round $r$ is at time \atreplace{$t$}{$\tEnter$} after GST, then the block proposed by $\leader_r$ will be committed by all correct replicas.
\end{lemma}
\begin{proof}
    Let $t_{QC}$ be the first time at which a correct replica forms (\cref{line:recv-qc-vote:form-qc}) or receives a valid QC in round $r$.
    Let us first establish several straightforward facts that can be verified by inspecting the pseudocode:
    \begin{enumerate}
        \item\label{itm:liv:no-timeouts} Assuming that the $\tRoundTimeout$ timer is set to at least $(4+\epsilon)\Delta$ (after correcting for any possible clock speed differences), no correct replica will time out (\cref{line:round-timeout:start}) round $r$ (or higher) before time $\atreplace{t}{\tEnter}+(4+\epsilon)\Delta$.
        
        \item\label{itm:liv:no-tc} Hence, no TC for round $r$ (or higher) will be formed before time $\atreplace{t}{\tEnter}+(4+\epsilon)\Delta$.

        \item\label{itm:liv:no-early-round-exit} Hence, no valid entry reason for round $r+1$ can be formed before the time $\min\{t_{QC}, \atreplace{t}{\tEnter}+(4+\epsilon)\Delta\}$.

        \item\label{itm:liv:no-rogue-qc-votes} No correct replica will QC-vote in round $r$ for any block other than the one proposed by $\leader_r$.

        \item\label{itm:liv:no-rogue-qc} Hence, no QC for any block other than the one proposed by $\leader_r$ will be formed in round $r$ (\cref{line:recv-qc-vote:form-qc}), and no correct replica will CC-vote for any other block (\cref{line:on-new-qc:cc-vote}).
    \end{enumerate}

    Suppose $t_{QC} \le \atreplace{t}{\tEnter}+(3+\epsilon)\Delta$.
    Then, by the time $t_{QC}+\Delta \le \atreplace{t}{\tEnter}+(4+\epsilon)\Delta$, every correct replica will have received the CC-vote containing the QC (\cref{line:recv-cc-vote}) and, as it could not have timed out by that time (fact~(\ref{itm:liv:no-timeouts})), issued a CC-vote in round $r$ (\cref{line:on-new-qc:cc-vote}).
    By the time $t_{QC}+2\Delta \le \atreplace{t}{\tEnter}+(5+\epsilon)\Delta$, every correct replica will have received the CC-votes from every other correct replica.
    Since all the CC-votes will have the same block hash (fact~(\ref{itm:liv:no-rogue-qc})), every correct replica will have formed a CC (\cref{line:recv-cc-vote:if}), committing the block (\cref{line:recv-cc-vote:commit-qc}).

    All that is left to prove is that, indeed, $t_{QC} \le \atreplace{t}{\tEnter}+(3+\epsilon)\Delta$.
    For the sake of contradiction, suppose the contrary, i.e., that no valid QC is formed or received by correct replicas by the time 
    $\atreplace{t}{\tEnter}+(3+\epsilon)\Delta$.
    In this case:

    \begin{description}
        \item[by $\atreplace{t}{\tEnter}+\Delta$:] $\leader_r$ will have received a valid $\mAdvanceRound$ (\cref{line:advance-round:start}) message and entered round $r$.

        \item[by $\atreplace{t}{\tEnter}+2\Delta$:] Every correct replica will have received a valid $\mPropose$ message from $\leader_r$ and started the $\tQCVote$ timer (\cref{line:recv-proposal:start-timer}). 

        \item[by $\atreplace{t}{\tEnter}+(2+\epsilon)\Delta$:] The $\tQCVote$ timer will have expired at every correct replica and, as no correct replica could have yet entered a higher round (fact~(\ref{itm:liv:no-early-round-exit})), every correct replica will have QC-voted in round $r$ (\cref{line:qc-vote-call-timer,line:qc-vote:sign-and-send}).

        \item[by $\atreplace{t}{\tEnter}+(3+\epsilon)\Delta$:] Every correct replica will have received $\mQCVote$ messages from every other correct replica and, as QC-votes will have the same block hash (fact~(\ref{itm:liv:no-rogue-qc-votes})), formed a QC unless it had already received one (\cref{line:recv-qc-vote:form-qc})---A contradiction.
    \end{description}
\end{proof}

\begin{atreview}
\begin{lemma} \label{lem:consecutive-rounds}
    If any correct replica enters round $r > 1$, then at least some correct replica entered round $r-1$.
\end{lemma}
\begin{proof}
    Recall that there are 3 types of entry reasons: $\rFullQC$, $\rCC$, and $\rTC$.
    $\rFullQC$ and $\rCC$ require a QC from the previous round, which, in turn, requires a quorum of QC-votes in the previous round.
    Similarly $\rTC$ requires a TC from the previous round, which requires a quorum TC-votes.
    A correct replica only issues QC-votes and TC-votes in its current round and each quorum must contain at least one correct replica.
\end{proof}

\begin{lemma} \label{lem:enter-inf-rounds}
    Correct replicas will enter an infinite number of rounds. Of them, an infinite number will have honest leaders.
\end{lemma}
\begin{proof}
    For the sake of contradiction, suppose that there is some round $\rLast$ such that some correct replica $p$ enters $\rLast$, but no correct replica ever enters any round $r > \rLast$.
    In this case, $p$ will eventually issue a TC-vote (\cref{line:round-timeout:tc-vote}).
    Since TC-votes include the entry reason, every other correct replica will eventually enter round $\rLast$ (\cref{line:recv-tc-vote:catch-up-round}) and also issue a TC-vote.
    Hence, all correct replicas will eventually receive enough TC-votes to form a TC and advance to round $\rLast+1$ (\cref{line:recv-tc-vote:advance-round})---A contradiction.

    The second part of the statement follows directly from \Cref{lem:consecutive-rounds}.
\end{proof}

\begin{lemma} \label{lem:commit-inf-blocks}
    Correct replicas will commit an infinite number of blocks.
    Of them, an infinite number proposed by honest replicas.
\end{lemma}
\begin{proof}
    Follows directly from \Cref{lem:enter-inf-rounds,lem:correct-leaders-commit}.
\end{proof}

\begin{lemma}[Quorum Store liveness] \label{lem:qs-liveness}
    If an honest replica invokes $\aBcast(m, \sn)$, then eventually every honest replica will receive a proof of availability for $(m, \sn)$.
\end{lemma}
\begin{proof}
    Trivially follows from the code in \Cref{alg:raptr:quorum-store}.
\end{proof}

\begin{theorem}
    {\Raptr} satisfies the validity and totality properties of BAB (\Cref{def:bab}).
\end{theorem}
\begin{proof}
    Validity follows directly from \Cref{lem:qs-liveness,lem:commit-inf-blocks}.

    Totality follows directly from \Cref{lem:commit-inf-blocks,lem:all-committed-qcs-are-compatible}.
\end{proof}
\end{atreview}

\end{document}